\documentclass[a4paper,UKenglish]{lipics}
\usepackage[utf8]{inputenc}
 
\usepackage{microtype}

\usepackage{mathtools}
\usepackage[ruled]{algorithm2e}
\usepackage{booktabs}

\newtheorem{claim}{Claim}

\newcommand{\abs}[1]{ \left| #1 \right|}
\newcommand{\cset}[2]{\left\{ #1  \, : \, #2\right\}}

\newcommand{\set}[1]{\left\{ #1 \right\}}

\newcommand{\eset}{{\mathcal E}}

\newcommand{\lset}{{\mathcal L}}
\renewcommand{\S}{{\mathcal S}}
\newcommand{\V}{{\mathcal V}}
\newcommand{\E}{{\mathcal E}}
\newcommand{\X}{{\mathcal X}}
\newcommand{\C}{\text{\textsc{Cong}}}
\renewcommand{\setminus}{-}

\title{New Approximability Results for the Robust $k$-Median Problem}

\author{Sayan Bhattacharya}
\author{Parinya Chalermsook}
\author{Kurt Mehlhorn}
\author{Adrian Neumann}
\affil{Max-Planck Institut für Informatik\\
Campus E1 4, \\
66123 Saarbrücken, Germany\\
\texttt{\{bsayan,parinya,mehlhorn,aneumann\}@mpi-inf.mpg.de}}
\authorrunning{S.\ Bhattacharya, P.\ Chalermsook, K.\ Mehlhorn, and A.\ Neumann }
\Copyright{S.\ Bhattacharya, P.\ Chalermsook, K.\ Mehlhorn, and A.\ Neumann }


\def\ShowComment{True}

\ifdefined\ShowComment
\def\parinya#1{\marginpar{$\leftarrow$\fbox{P}}\footnote{$\Rightarrow$~{\sf #1 --Parinya}}}

\else
\def\parinya#1{}
\fi

\begin{document}

\maketitle

\begin{abstract}
We consider a robust variant of the classical $k$-median problem, introduced by Anthony et al.~\cite{AnthonyGGN10}.
In the \emph{Robust $k$-Median problem}, we are given an $n$-vertex metric space $(V,d)$ and $m$ client sets $\set{S_i \subseteq V}_{i=1}^m$. The objective is to open a set $F \subseteq V$ of $k$ facilities such that the worst case connection cost over all client sets is minimized; in other words, minimize $\max_{i} \sum_{v \in S_i} d(F,v)$. Anthony et al.\ showed an $O(\log m)$ approximation algorithm for any metric and APX-hardness even in the case of uniform metric. In this paper, we show that their algorithm is nearly tight by providing $\Omega(\log m/ \log \log m)$ approximation hardness, unless ${\sf NP} \subseteq \bigcap_{\delta >0} {\sf DTIME}(2^{n^{\delta}})$. This hardness result holds even for uniform and line metrics. To  our knowledge, this is one of the rare cases in which a problem on a line metric is hard to approximate to within logarithmic factor. We complement the hardness result by an experimental evaluation of different heuristics that shows that very simple heuristics achieve good approximations for realistic classes of instances.
\end{abstract}

\section{Introduction}

In the classical $k$-median problem, we are given a set of  clients  located on a metric space with distance function $d : V \times V \rightarrow \mathbb{R}$. The goal is  to open a set of facilities $F \subseteq V$,  $|F| = k$, so as  to minimize the sum of the connection costs of the clients in $V$, i.e., their  distances from their  nearest facilities in $F$.  This is a central problem in approximation algorithms, and quite naturally, it has received a large amount of attention in the past two decades~\cite{CharikarG99,AryaGKMMP04,CharikarGTS02,LinV92,LiS13}. 
 
At SODA 2008 Anthony et al.~\cite{anthony2008plant,AnthonyGGN10} introduced a generalization of the $k$-median problem. In their setting, the set of clients that are to be  connected to some facility is not known in advance, and the goal is to perform well in spite of this uncertainty about the future. 
In particular, they formulated the problem as follows.

\begin{definition}[Robust $k$-Median]
\label{def:mainproblem}
An instance of this problem is a triple $(V, \S, d)$. This defines a set of \emph{locations} $V$, a collection of $m$ sets of \emph{clients} $\S = \{S_1,\ldots , S_m\}$, where $S_i \subseteq V$ for  all $i \in \{1,\ldots, m\}$, and a metric distance function $d: V \times V \rightarrow \mathbb{R}$. We have to open a set of $k$ facilities $F \subseteq V$, $|F| = k$, and the goal is to minimize the cost of the most expensive set of clients, i.e.\ minimize $\max_{i=1}^m \sum_{v \in S_i} d(v, F )$. Here, $d(v,F)$ denotes the minimum distance of the client $v$ from any location in $F$, i.e.\ $d(v,F) = \min_{u \in F} d(u,v)$.
\end{definition}

Note that the Robust $k$-Median problem is a natural generalization of the classical $k$-median problem (where $m=1$). In addition, we can think of  this formulation as capturing some notion of {\em fairness}. To see this, simply  interpret each set $S_i$ as  a {\em community} of clients who would pay $\sum_{v \in S_i} d(v,F)$ for getting connected to some facility. Now the objective ensures that no single community pays too much, while minimizing the cost.  
Anthony et al.~\cite{AnthonyGGN10} gave an $O(\log m)$-approximation algorithm for this problem, and a lower bound of $(2-\epsilon)$ for the best possible approximation ratio by a reduction from Vertex Cover.


\paragraph*{Our Results}
We give  nearly tight hardness of approximation results for the Robust $k$-Median problem. 
We show that  unless ${\sf NP} \subseteq \cap_{\delta >0} {\sf DTIME}(2^{n^{\delta}})$, the problem admits  no poly-time $o(\log m/\log \log m)$-approximation, {\em even on uniform and line metrics}.

Our first hardness result  is tight up to a constant factor, as a simple  rounding scheme  gives a matching upper bound on uniform metrics (see Section~\ref{sec:gap}). Our second, and rather surprising, result shows  that  ``Robust $k$-Median''  is a rare problem with  super-constant hardness of approximation even on line metrics, in sharp contrast to most other geometric optimization problems which admit polynomial time approximation schemes, e.g.~\cite{arora1996polynomial, kolliopoulos1999nearly}. 

In Section~\ref{sec:experiments} we investigate the performance of some heuristics. Already a very simple greedy strategy provides reasonably good performance on a realistic class of instances. We use an LP relaxation of the problem as a lower bound.

\paragraph*{Our Techniques} 
First, we note that the Robust $k$-Median problem on uniform metrics is equivalent to the following variant of the set cover problem: Given a set $U$ of ground elements, a collection of sets $\X=\set{X \subseteq U}$, and an integer $t \leq |\X|$, our goal is to select $t$ sets from $\X$ in order to minimize the number of times an element from $U$ is hit (see Lemma~\ref{lm:equivalence}). We call this problem Minimum Congestion Set Packing (MCSP). 
This characterization allows us to focus on proving the hardness of MCSP, and to employ the tools developed for the set cover problem. 

We now revisit the reduction used in proving the hardness of the set cover problem by Feige~\cite{feige}, building on the framework of Lund and Yannakakis~\cite{LundY94}, and discuss how our approach differs from theirs. Intuitively, they compose the Label Cover instance with a set system that has some desirable properties. Informally speaking, in the Label Cover problem, we are given a graph  where each vertex $v$ can be assigned a label from a set $L$, and each edge $e$ is equipped with a constraint $\Pi_e \subseteq L \times L$ specifying the accepting pairs of labels for $e$. Our goal is to find a labeling of vertices that maximizes the number of accepting edges. This problem is known to be hard to approximate to within a factor of $2^{\log^{1-\epsilon} |E|}$~\cite{AroraLMSS98,Raz98}, where $|E|$ is the number of edges. Thus, if we manage to reduce Label Cover to MCSP, we would hopefully obtain a large hardness of approximation factor for MCSP as well.

From the Label Cover instance, \cite{LundY94} creates an instance of Set Cover by having sets of the form $S(v, \ell)$ for each vertex $v$ and each label $\ell \in L$. 
Intuitively the set $S(v,\ell)$ means choosing label $\ell$ for vertex $v$ in the label cover instance.
Now, if we assume that the solution is well behaved, in the sense that for each vertex $v$, only one set of the form $S(v,\ell)$ is chosen in the solution, we would be immediately done (because each set indeed corresponds to label assignment). 
However, a solution need not have this form, e.g. choosing sets $S(v, \ell)$ and $S(v, \ell')$ would translate to having two labels $\ell, \ell'$ for the label cover instance. 
To prevent an ill-behaved solution, ``partition systems'' were introduced and used in both~\cite{LundY94} and~\cite{feige}. 
Feige considers the hypergraph version of Label Cover to obtain a sharper hardness result of $\ln n - O(\ln \ln n)$ instead of $\frac{1}{4} \ln n$ in~\cite{LundY94}; here $n$ denotes the size of the universe.
 
Now we highlight how our reduction is different from theirs. The high level idea of our reduction is the same, i.e. we have sets of the form $S(v,\ell)$ that represent assigning label $\ell$ to vertex $v$. 
However, we need a different partition system and a totally different analysis.
Moreover, while a reduction from standard Label Cover gives nearly tight $O(\log n)$ hardness for Set Cover, it can (at best) only give the hardness of $2-\epsilon$ for MCSP. 
To prove our results, we do need a reduction from the Hypergraph Label Cover problem. 
This suggests another natural distinction between MCSP and Set Cover.   

Finally, to obtain the hardness of the Robust $k$-Median problem on the line metric, we embed the instance created from the MCSP reduction onto the line such that the values of optimal solutions are preserved. This way we get the same hardness gap for line metrics.




\section{Preliminaries}

We will show that the Robust $k$-Median problem is $\Omega(\log m / \log \log m)$ hard to approximate, even for the special cases of \emph{uniform metrics} (see Section~\ref{subsec:hard:uniform}) and \emph{line metrics} (see Section~\ref{subsec:hard:line}). Recall that $d$ is a uniform metric iff we have $d(u,v) \in \{0,1\}$ for all locations $u,v \in V$. Further, $d$ is a line metric iff the locations in $V$ can be embedded into a line in such a way that $d(u,v)$ equals the Euclidean distance between $u$ and $v$, for all $u,v \in V$.  Throughout this paper, we will denote any set  of the form $\{1,2,\ldots, i\}$ by $[i]$.    Our hardness results will rely on a reduction from the \emph{$r$-Hypergraph Label Cover} problem, which is defined as follows. 

\begin{definition}[$r$-Hypergraph Label Cover]
\label{def:labelcover}
An instance of this problem is a triple $(G,\pi,r)$, where $G = (\mathcal{V},\mathcal{E})$ is a $r$-partite hypergraph with vertex set $\V = \bigcup_{j=1}^r \V_j$ and edge set $\E$. Each edge $h \in \mathcal{E}$ contains one vertex from each part  of $\V$, i.e.\ $|h \cap \V_j| = 1$ for all $j \in [r]$. Every set $\V_j$ has an associated set of \emph{labels} $L_j$.  Further, for all $h \in \mathcal{E}$ and $j \in [r]$, there is a mapping $\pi_h^j : L_j \rightarrow C$ that projects the labels from $L_j$ to a common set of \emph{colors} $C$. 

The problem is to assign to every vertex $v \in \V_j$ some label $\sigma(v) \in L_j$. We say that an edge $h = (v_1,\ldots,v_r)$, where $v_j \in \V_j$ for all $j \in [r]$, is \emph{strongly satisfied} under $\sigma$ iff the labels of all its vertices are mapped to the same element in $C$, i.e.   $\pi_h^j(\sigma(v_j)) = \pi_h^{j'} (\sigma(v_{j'}))$ for all $j, j' \in [r]$. In contrast, we say that the edge is \emph{weakly satisfied} 
iff there exists some pair of vertices in $h$ whose labels are mapped to the same element in $C$, i.e.\ $\pi^j_h(\sigma(v_j)) = \pi^{j'}_h(\sigma(v_{j'}))$ for some $j, j' \in [r]$, $j \neq j'$. 
\end{definition}

For ease of exposition, we will often abuse the notation and denote by $j(v)$ the part of $\V$ to which a vertex $v$ belongs, i.e. if $v \in \V_{j}$ for some $j \in [r]$, then we set $j(v) \leftarrow j$. The next theorem will be crucial in deriving our hardness result. The proof of this theorem follows from Feige's $r$-Prover system~\cite{feige} (see Appendix~\ref{sec:labelcover}).

\begin{theorem}
\label{th:labelcover}
Let $r \in \mathbb{N}$ be a parameter. There is a polynomial time reduction from $n$-variable 3-SAT to 
 $r$-Hypergraph Label Cover with the following properties: 
\begin{itemize}
\item (Yes-Instance) If the formula is satisfiable, then there is a labeling that strongly satisfies every edge in $G$. 
\item (No-Instance) If the formula is not satisfiable, then every labeling weakly satisfies at most a $2^{-\gamma r}$ fraction of the edges in $G$, for some universal constant $\gamma$. 
\item The number of vertices in the graph is $|\V| = n^{O(r)}$ and the number of edges is $|\mathcal{E}| = n^{O(r)}$. The sizes of the label sets are $\abs{L_j}= 2^{O(r)}$ for all $j \in [r]$, and $\abs{C} = 2^{O(r)}$. Further, we have $|\V_j| = |\V_{j'}|$ for all $j, j' \in [r]$, and each vertex $v \in \V$ has the same degree $r |\mathcal{E}|/|\V|$.
\end{itemize}
\end{theorem}

We use a \emph{partition system} that is motivated by the hardness proof of the Set Cover problem~\cite{feige}. However, we deal with a different problem, and our construction is also different.

\begin{definition}[Partition System]
\label{def:partition} Let $r \in \mathbb{N}$ and let $C$ be any finite set. 
An $(r,C)$-partition system is a pair $(Z, \set{p_{c}}_{c \in C})$, where $Z$ is an arbitrary (ground) set, and for each $c \in C$, $p_{c}$ is a partition of $Z$ into $r$ subsets, such that the following properties hold. 

\begin{itemize} 
\item (Partition) For all $c \in C$, $p_c = \left(A_c^1,\ldots, A_c^r\right)$ is a partition of $Z$, that is $\bigcup_{j=1}^r A_c^j = Z$, and $A_c^{j'} \cap A_c^{j} = \emptyset$ for all $j, j' \in [r], j \neq j'$. 
\item ($r$-intersecting) For any $r$ \emph{distinct} indices $c_1,\ldots, c_r \in C$ and \emph{not-necessarily distinct} indices $j_1,\ldots, j_r \in [r]$, we have that $\bigcap_{i=1}^r A_{c_i}^{j_i} \neq \emptyset$. In particular, $A_c^j \not= \emptyset$ for all $c$ and $j$. 
\end{itemize} 
\end{definition}

In order to achieve a good lower bound on the approximation factor, we need partition systems with \emph{small} ground sets. The most obvious way to build a partition system is to form an $r$-hypercube: Let $Z= [r]^{|C|}$, and for each $c \in C$ and $j \in [r]$, let $A_c^j$ be the set of all elements in $Z$ whose $c$-th component is $j$. It can easily be verified that this is an $(r,C)$-partition system with $|Z| = r^{\abs{C}}$. With this construction, however, we would only get a hardness of $\Omega(\log \log m)$ for our problem. The following lemma shows that it is possible to construct an $(r,C)$-partition system probabilistically with $|Z| = r^{O(r)}\log \abs{C}$. 

\begin{lemma}
\label{lm:partition}
There is an $(r,C)$-partition system with $|Z| = r^{O(r)} \log \abs{C}$ elements. 
Further, such a partition system can be constructed efficiently with high probability. 
\end{lemma} 

\begin{proof}
Let $Z$ be any set of $r^{O(r)} \log \abs{C}$ elements. 
We build a partition system $(Z,\{p_c\}_{c \in C})$ as described in Algorithm~\ref{alg:partition}.

\begin{algorithm}[h]
\DontPrintSemicolon
\SetAlgoNoEnd
\caption{A randomized algorithm for constructing an $(r,C)$-partition system.}\label{alg:partition}
\SetKwInOut{Input}{input}\SetKwInOut{Output}{output}
\Input{A ground set $Z$, a parameters $r \in \mathbb{N}$, and a set $C$.}
\ForEach{$c \in C$}{\tcc{Construct the partition $p_c=(A_c^1,\ldots,A_c^r)$}
    Initialize $A_c^j$ to the empty set for all $j \in [r]$\;
  \ForEach{ground element $e\in Z$}{
    Pick an index $j\in[r]$ independently and uniformly at random and add $e$ to $A_c^j$
  }
}
\end{algorithm}

In Algorithm~\ref{alg:partition}, by construction each $p_c$ is a partition of $Z$, i.e.\ the first property stated in Definition~\ref{def:partition} is satisfied. We bound the probability that the second property is violated.

Fix any choice of $r$ \emph{distinct} indices $c_1, \ldots, c_r \in C$ and \emph{not necessarily distinct} indices $j_1, \ldots, j_r \in [r]$. We say that a \emph{bad event} occurs when the intersection of the corresponding sets is empty, i.e.\ $\bigcap_{i=1}^r A_{c_i}^{j_i} = \emptyset$. To upper bound the probability of a bad event, we focus on events of the form $E_{e, i}$ -- this occurs when an element $e \in Z$ is included in a set $A_{c_i}^{j_i}$. Since the indices $c_1 \ldots c_r$ are distinct, it follows that the events $\{E_{e,i}\}$ are mutually independent. Furthermore, note that we have $\Pr[E_{e,i}] = 1/r$ for all $e \in Z, i \in [r]$. Hence, the probability that an element $e \in Z$ does not belong to the intersection $\bigcap_{i=1}^r A_{c_i}^{j_i}$ is given by $1 - \Pr[\bigcap_{i=1}^r E_{e,i}] = 1-1/r^r$. Accordingly, the probability that no element $e \in Z$ belongs to the intersection, which defines the bad event, is equal to $(1-1/r^r)^{|Z|}$.

Now, the total number of choices for $r$ distinct indices $c_1, \ldots, c_r$ and $r$ not-necessarily distinct indices $j_1, \ldots, j_r$ is equal to $\binom{\abs{C}}{r} \cdot r^r$. Hence, taking a union-bound over all possible bad events, we see that the second property stated in Definition~\ref{def:partition} is violated with probability at most $\binom{\abs{C}}{r} \cdot r^r \cdot (1-r^r)^{|Z|} \leq (\abs{C}r)^r \cdot \exp(-|Z|/r^r)$. If we set $|Z| = d \cdot r^{d \cdot r} \log \abs{C}$ with sufficiently large constant $d$, then it is easy to see that the second constraint in Definition~\ref{def:partition} is satisfied with high probability.
\end{proof}




\section{Hardness of Robust k-Median on Uniform Metrics}
\label{subsec:hard:uniform}

First, we define a problem called \emph{Minimum Congestion Set Packing} (MCSP), and then show a reduction from MCSP to Robust $k$-Median on uniform metrics. In Section~\ref{subsec:reduction}, we will then show that MCSP is hard to approximate by reducing Hypergraph Label Cover to MCSP. 

\begin{definition}
\label{def:newproblem}
[Minimum Congestion Set Packing (MCSP)] An instance of this problem is a triple $(U, \X, t)$, where $U$ is a universe of $m$ elements, i.e.\ $|U| = m$, $\X$ is a collection of sets $\X=\set{X \subseteq U}$ such that $\bigcup_{X \in \X} X = U$, and $t \in \mathbb{N}$ and $t \le \abs{\X}$. The objective is to find a collection $\X' \subseteq \X$ of size $t$ that minimizes $\C(\X') = \max_{e\in U} \C(e, \X')$. Here, $\C(\X')$ refers to the \emph{congestion} of the solution $\X'$, and $\C(e,\X')= |\{ X \in \X' : e \in X\}|$ is the congestion of the element $e \in U$ under the solution $\X'$.
\end{definition}


\begin{lemma}
\label{lm:equivalence}
Given any MCSP instance $(U,\X,t)$, we can construct a Robust $k$-Median instance $(V,\S,d)$ with the same objective value in $poly(|U|, |\X|)$ time, such that $|U| = |\S|$, $|\X| = |V|$, $d$ is a uniform metric, and $k = |V| - t$.
\end{lemma}

\begin{proof}
We construct the Robust $k$-Median instance $(V,\S,d)$ as follows. For every $e \in U$ we create a set of clients $S(e)$, and for each $X \in \X$ we create a location $v(X)$. Thus, we get $V = \{ v(X) \, : \, X \in \X\}$, and $\S = \{ S(e) \, : \, e \in U\}$. We place the clients in $S(e)$ at the locations of the sets that contain $e$, i.e.\ $S(e) = \{ v(X) \, : \, X \in \X, e \in X\}$ for all $e \in U$. The distance is defined as $d(u,v) = 1$ for all $u, v \in V, u \neq v$, and $d(v,v) = 0$. Finally, we set $k \leftarrow |V| - t$.

Now, it is easy to verify that the Robust $k$-Median instance $(V,\S,d)$ has a solution with objective $\rho$ iff the corresponding MCSP instance $(U,\X, t)$ has a solution with objective $\rho$. The intuition is that a location $v(X) \in V$ is \emph{not} included in the solution $F$ to the Robust $k$-Median instance iff the corresponding set $X$ is included in the solution $\X'$ to the MCSP instance. Indeed, let $F$ be any subset of $\X$ of size $k$ (= the set of open facilities) and let $\X' = \X \setminus F$. Further, let  $[X \in \X']$ be an indicator variable that is set to $1$ iff $X \in \X'$. Then
\begin{align*}
\C(\X') &= \max_{e \in U} \C(e,\X') = \max_{e \in U} \sum_{X; e \in X} [ X \in \X'] \\
&= \max_{e \in U} \sum_{X; e \in X} \min_{Y \in F} d(X,Y) = \max_{S(e) \in \S} \sum_{v(X) \in S(e)} d(v(X),F).
\end{align*}
\end{proof}

We devote the rest of Section~\ref{subsec:hard:uniform} to the MCSP problem and show that it is $\Omega(\log |U|/\log \log |U|)$ hard to approximate. This, in turn, will imply a $\Omega(\log |\S| / \log \log |\S|)$ hardness of approximation for Robust $k$-Median on uniform metrics. We will prove the hardness result via a reduction from Hypergraph Label Cover.



\subsection{Integrality Gap}
\label{sec:gap} 

Before proceeding to the hardness result, we show that a natural LP relaxation for the MCSP problem~\cite{AnthonyGGN10} has an integrality gap of $\Omega(\log m/ \log \log m)$, where $m = |U|$ is the size of the universe of elements. In the LP, we have a variable $y(X)$ indicating that the set $X \in \X$ is chosen, and a variable $z$ which represents the maximum congestion among the elements. 
\begin{align*}
  \min\quad & z\\
 \mbox{s.t.}\quad & \sum_{\mathclap{X \in \X: e \in X}} y(X) \leq z \mbox{ for all $e \in U$} \\
  & \sum_{\mathclap{X \in \X}} y(X) = t 
\end{align*} 

\noindent \textbf{The Instance:} Now, we construct a bad integrality gap instance $(U,\X,t)$. 
Let $d$ be the intended integrality gap, let $\eta = d^2$, and let $U = \cset{I}{I \subseteq [\eta], \abs{I} = d}$ be all subsets of $[\eta]$ of size $d$. The collection $\X$ consists of $\eta$ sets
$X_1,\ldots, X_{\eta}$, where $X_i = \cset{I}{I \in U \text{ and } i \in I}$. Note that the universe $U$ consists of $|U| = m = \binom{\eta}{d}$ elements, and each element $I$ is contained in exactly $d$ sets, namely $I \in X_i$ if and only if $i \in I$.
Finally, we set $t \leftarrow \eta/d$. 

\smallskip
\noindent \textbf{Analysis:}
The fractional solution simply assigns a value of $1/d$ to each variable $y(X_i)$; this ensures that the total (fractional) number of sets selected is $\eta/d = t$.
Furthermore, each element is contained (fractionally) in exactly one set, so the fractional solution has cost one. 
Any integral solution must choose $\eta/d = d$ sets, say $X_{i_1} \ldots X_{i_d}$. Then $I = \set{i_1,\ldots,i_d} \in X_{i_{\lambda}}$ for all $\lambda \in [d]$ and hence the congestion of $I$ is $d$, and this also means that any integral solution has cost at least $d$. 
Finally, since $|U| = m \leq \eta^d \leq (d^2)^d$, we have $d = \Omega(\log m/ \log \log m)$. 



\smallskip
\noindent \textbf{Tightness of the result:}
The bound on the hardness and integrality gap is tight for the uniform metric case, as there is a simple $O(\log m/ \log \log m)$-approximation algorithm. Pick each set $X$ with probability equal to $\min(1,2y(X))$. The expected congestion is $2z$ for each element. By Chernoff's bound~\cite{Hagerup-Rueb}, an element is covered by no more than $z \cdot O(\log m/ \log \log m)$ sets with high probability. A similar algorithm gives the same approximation guarantee for the Robust $k$-Median problem on uniform metrics.

\subsection{Reduction from r-Hypergraph Label Cover to Minimum Congestion Set Packing} 
\label{subsec:reduction}
The input is an instance $\left(G,\pi,r\right) $ of the $r$-Hypergraph Label Cover problem (see Definition~\ref{def:labelcover}). From this we construct the following instance  $(U, \X, t)$ of the MCSP problem (see Definition~\ref{def:newproblem}).
\begin{itemize} 
\item First, we define the universe $U$ as a union of disjoint sets. For each edge $h \in \E$ in the hypergraph we have a set $U_h$. All  these sets have the same size $m^*$ and are pairwise disjoint, i.e.\  $U_h \cap U_{h'} = \emptyset$ for all $h, h' \in \eset$, $h' \neq h$. The universe $U$ is then the union of these sets $U = \bigcup_{h \in \E} U_h$. Since the $U_h$ are mutually disjoint, we have $m=|U|=|\eset|\cdot m^*$. 
Recall that $C$ is the target set of $\pi$. Each set $U_h$ is the ground set of an $(r,C)$-partition system (see Definition~\ref{def:partition}) as given by Lemma~\ref{lm:partition}. In particular we have  $m^* = r^{O(r)} \log \abs{C}$. We denote the $r$-partitions associated with  $U_h$ by $\{p_c(h)\}_{c \in C}$, where $p_c(h) = \left(A^1_c(h), \ldots, A^r_c(h)\right)$.
\item Second, we construct the collection of sets $\X$ as follows. For each $j \in [r]$, $v \in \V_j$ and $\ell \in L_j$, $\X$ contains the set $X(v,\ell)$, where $X(v,\ell) = \bigcup_{h: v \in h} A^j_{\pi^j_h(\ell)}(h)$. In other words, $X(v,\ell) \cap U_h$ is empty if $v \not\in h$ and is equal to $A^j_{\pi^j_h(\ell)}(h)$ if $v \in h$. 
Intuitively, choosing the set $X(v,\ell)$ corresponds to assigning label $\ell$ to the vertex $v$. 
\item Third, we define $t \leftarrow |\V|$. Intuitively, this means that each vertex in $\V$ gets one label.
\end{itemize} 

 We assume for the sequel that the $r$-Hypergraph Label Cover instance is chosen according to Theorem~\ref{th:labelcover}. We assume that the parameter $r$ satisfies $r^7 2^{-\gamma r} < 1$. In the proof of the main theorem, we will fix $r$ to a specific value.

\subsection{Analysis}

We show that the reduction from Hypergraph Label Cover to MCSP satisfies two properties. In Lemma~\ref{lm:YES}, we show that for a Yes-Instance (see Theorem~\ref{th:labelcover}), the corresponding MCSP instance admits a solution with congestion one. Second, in case of a No-Instance, we show in Lemma~\ref{lm:NO} that every solution to the corresponding MCSP instance has congestion at least $r$.

\begin{lemma}[Yes-Instance]
\label{lm:YES}
If the Hypergraph Label Cover instance $(G,\pi,r)$ admits a labeling that strongly satisfies every edge, then the MCSP instance $(U,\X,t)$ constructed in Section~\ref{subsec:reduction} admits a solution where the congestion of every element in $U$ is exactly one.
\end{lemma}

\begin{proof}
Suppose that there is a labeling $\sigma$ that strongly satisfies every edge $h \in \eset$.
We will show how to pick $t = |\V|$ sets from $\X$ such that each element in $U$ is contained in exactly one set. This implies that the maximum congestion is one. For each $j \in [r]$ and each vertex $v \in \V_j$, we choose the set $X(v, \sigma(v))$. Thus, the total number of sets chosen is exactly $|\V|$. 

To see that the congestion is indeed one, we concentrate on the elements in $U_h$, where $h = (v_1,\ldots, v_r)$, $v_j \in \V_j$ for all $j \in [r]$, is one of the edges in $\E$. The picked sets that intersect $U_h$ are $X(v_j,\sigma(v_j))$, where $j \in [r]$. Since $h$ is strongly satisfied, $\pi_h$ maps all labels of the vertices in $h$ to a common $c \in C$, that is $\pi^j_h(\sigma(v_j)) = c$ for all $j \in [r]$. Thus  $U_h \cap X(v_j, \sigma(v_j)) = A^j_c(h)$. By the definition of a partition system (see Definition~\ref{def:partition}), the sets $A^1_c(h) \ldots A^r_c(h)$ partition the elements in $U_h$. This completes the proof.
\end{proof}

Now, we turn to the proof of Lemma~\ref{lm:NO}. Towards this end, we fix a collection $\X' \subseteq \X$ of size $t$ and show that some element in $U$ has congestion at least $r$ under $\X'$. The intuition being that many edges in $G = (\V, \E)$ are not even weakly satisfied, and the elements in $U$ corresponding to those edges incur large congestion. Recall that for a $v \in \V$, we define $j(v) \in \mathbb{N}$ to be such that $v \in \V_{j(v)}$. 

\begin{claim}
\label{cl:1}
For $v \in \V$, let $\lset_v = \cset{\ell \in L_{j(v)}}{X(v,\ell) \in \X'}$. For $h \in \E$, let $\Lambda_h = \cset{X(v,\ell) \in \X'}{v \in h}$ and $\lambda(h) = \abs{\Lambda_h}$. If the solution $\X'$ has congestion less than $r$ then $\abs{\lset_v} < r^2$ and $\abs{\Lambda_h} < r^3$. 
\end{claim}
\begin{proof} Since $\Lambda_h = \bigcup_{v \in h} \lset_v$, it suffices to prove $\abs{\lset_v} < r^2$ for all $v$. Assume otherwise, i.e., $\abs{\lset_v} \ge r^2$ for some $v \in \V_j$, $j \in [r]$. Let $h$ be any hyper-edge with $v \in h$. Consider the images of the labels in $\lset_v$ under the projection $\pi^j_h$. Either we have at least $r$ distinct images or at least $r$ elements in $L_v$ are mapped to the same element of $C$. 

In the former case, we have $r$ pairwise distinct labels $\ell_1$ to $\ell_r$ in $\lset_v$ and $r$ pairwise distinct labels $c_1$ to $c_r$ in $C$ such that $\pi^j_h(\ell_i) = c_i$ for $i \in [r]$. The set $X(v,\ell_i)$ contains $A^j_{c_i}(h)$ and $\bigcap_{i \in [r]} A^j_{c_i}(h) \not= \emptyset$ by property (2) of partition systems (see Definition~\ref{def:partition}). Thus some element has congestion at least $r$.

In the latter case, we have $r$ pairwise distinct labels $\ell_1$ to $\ell_r$ in $\lset_v$ and a label $c$ in $C$ such that $\pi^j_h(\ell_i) = c$ for $i \in [r]$. The set $X(v,\ell_i)$ contains $A^j_c(h)$ and hence every element in this non-empty set (property (2) of partition systems) has congestion at least $r$.
\end{proof}

\begin{definition}[Colliding Edge]
\label{def:colliding}
We say that an edge $h \in \E$ is \emph{colliding} iff there are sets $X(v,\ell), X(v',\ell') \in \X'$ with $v,v' \in h$, $v \not= v'$, and $\pi^{j(v)}_{h}(\ell) = \pi^{j(v')}_h(\ell')$. 
\end{definition}

\begin{claim}
\label{cl:colliding}
Suppose that the solution $\X'$ has congestion less than $r$, and more than a $r^4 2^{-\gamma r}$ fraction of the edges in $\E$ are colliding. Then there is a labeling $\sigma$ for $G$ that weakly satisfies at least  a $2^{-\gamma r}$ fraction of the edges in $\E$.
\end{claim}  

\begin{proof}
For each $v \in \V$, we define the label set $\lset_v = \{ \ell \in L_{j(v)} \, : \, X(v,\ell) \in \X'\}$. Then $\abs{\lset_v} < r^2$ by Claim~\ref{cl:1}. We construct a labeling function $\sigma$ using Algorithm~\ref{alg:labeling function}. 
\begin{algorithm}[h]
\DontPrintSemicolon
\SetAlgoNoEnd
\caption{An algorithm for constructing a labeling function.}\label{alg:labeling function}
\ForEach{vertex $v \in \V$}{
  \If{$\lset_v \neq \emptyset$}{
    Pick a color $\sigma(v)$ uniformly and independently at random from $\lset_v$\;
  } \Else {
    Pick an arbitrary color $\sigma(v)$ from $L_{j(v)}$\;
  }
}
\end{algorithm}

 Now we bound the expected fraction of weakly satisfied edges under $\sigma$ from below. Take any colliding edge $h \in \E$. This means that there are vertices $v \in \V_j$, $v' \in \V_{j'}$ with $j \neq j'$, and colors $\ell \in \lset_v$, $\ell' \in \lset_{v'}$ such that $v, v' \in h$ and $\pi^j_h(\ell) = \pi^{j'}_h(\ell')$. By Claim~\ref{cl:1}, $|\lset_v|$ and $|\lset_{v'}|$ are both at most $r^2$. Since the colors $\sigma(v)$ and $\sigma(v')$ are chosen uniformly and independently at random from their respective palettes $\lset_v$ and $\lset_{v'}$, we have $\Pr[\sigma(v) = \ell \text{ and } \sigma(v') = \ell'] \geq 1/r^4$. In other words, every colliding edge is weakly satisfied with probability at least $1/r^4$. Since more than a $r^4 2^{-\gamma r}$ fraction of the edges in $\E$ are colliding, from linearity of expectation we infer that the expected fraction of edges weakly satisfied by $\sigma$ is at least $2^{-\gamma r}$.
 \end{proof}

\begin{claim}
\label{cl:counting}
Let $\Lambda_h = \cset{X(v,\ell) \in \X'}{v \in h}$,  and $\lambda(h) = |\Lambda_h|$. We have $\sum_{h \in \eset} \lambda(h) = r |\eset|$.
\end{claim} 
\begin{proof}
This is a simple counting argument. Consider a bipartite graph $H$ with vertex set $A\dot \cup B$, where each vertex in $A$ represents a set $X(v,\ell)$, and each vertex in $B$ represents an edge $h \in \eset$. There is an edge between two vertices iff the set $X(v,\ell)$ contains some element in $U_h$. The quantity $\sum_{h \in \eset} \lambda(h)$ counts the number of edges in $H$ where one endpoint is included in the solution $\X'$. Since $\X'$ picks $t = |\V|$ sets and each set has degree $r |\eset|/|\V|$ in the $H$ (see Theorem~\ref{th:labelcover}), the total number of edges that are chosen is exactly $ |\V| \times \left(r |\eset|/|\V|\right) = r|\eset|$. 
\end{proof}

 Let $\eset' \subseteq \eset$ denote the set of colliding edges, and define $\eset'' = \eset \setminus \eset'$. Suppose that we are dealing with a No-Instance (see Theorem~\ref{th:labelcover}), i.e.\ the solution $\X'$ has congestion less than $r$ and every labeling weakly satisfies at most a $2^{-\gamma r}$ fraction of the edges in $\E$. Then $\lambda(h) \leq r^3$ for all $h \in \E$ by Claim~\ref{cl:1}, and no more than $r^4 2^{-\gamma r}|\E|$ edges are colliding, i.e.\ $\abs{\eset'} \le r^4 2^{-\gamma r}\abs{\eset}$, by Claim~\ref{cl:colliding}. Using these facts we conclude that $\sum_{h \in \E'} \lambda(h) \leq r^7 2^{-\gamma r} |\E| < |\E|$, as by assumption $r^72^{-\gamma r}<1$. Now, applying Claim~\ref{cl:counting}, we get $\sum_{h \in \E''} \lambda(h) = r|\E| - \sum_{h \in \E'} \lambda(h) > (r-1) |\E|$. In particular, there is an edge $h \in \E''$ with $\lambda(h) \geq r$.

\newcommand{\barX}{\Lambda_h}

Recall that $\Lambda_h = \cset{X(v,\ell) \in \X'}{v \in h}$ are the sets in $\X'$  that intersect $U_{h}$ and note that $|\barX| = \lambda(h)\geq r$. Let $\X^* \subseteq \barX$ be a \emph{maximal} collection of sets  with the following property: For every two distinct sets $X(v,\ell), X(v',\ell') \in \X^*$ we have  $\pi^{j(v)}_{h}(\ell) \neq \pi^{j(v')}_{h}(\ell')$. Hence, from the definition of a partition system (see Definition~\ref{def:partition}), it follows that the intersection of the sets in $\X^*$ and the set $U_{h}$ is nonempty.

Now, consider any set $X(v,\ell) \in \barX \setminus \X^*$. Since the collection $\X^*$ is maximal, there must be at least one set $X(v',\ell')$ in $\X^*$ with $\pi^{j(v)}_{h}(\ell) = \pi^{j(v')}_{h}(\ell')$. Since $h$ is not colliding, we must have $j(v) = j(v')$. Consequently we get $X(v,\ell) \cap U_{h} = X(v',\ell') \cap U_{h}$. In other words, for every set $X \in \barX \setminus \X^*$, there is some set $X' \in \X^*$ where $X \cap U_{h} = X' \cap U_{h}$. Thus, $U_h \cap (\bigcap_{X \in \barX} X) = U_h \cap (\bigcap_{X \in \X^*} X) \neq \emptyset$. Every element in the intersection of the sets in $\Lambda_h$ and $U_{h}$ will have congestion $|\Lambda_h| \geq r$. This leads to the following lemma.

\begin{lemma}[No-Instance]
\label{lm:NO}
Suppose that every labeling weakly satisfies at most a $2^{-\gamma r}$ fraction of the edges in the hypergragph label cover instance $(G,\pi,r)$, for some universal constant $\gamma$ and that $r^7 2^{-\gamma r} < 1$. Then the congestion incurred by every solution to the MCSP instance $(U,\X,t)$ constructed in Section~\ref{subsec:reduction} is at least $r$.
\end{lemma}

We are now ready to prove the main theorem of this section.

\begin{theorem}
\label{th:hard:uniform}
The Robust $k$-Median problem $(V,\S,d)$ is $\Omega(\log m / \log \log m)$ hard to approximate on uniform metrics, where $m = |\S|$, unless $\textsf{NP} \subseteq \bigcap_{\delta > 0} \textsf{DTIME} (2^{n^{\delta}})$.
\end{theorem}

\begin{proof} Assume that there is a polynomial time algorithm for the Robust $k$-Median problem that guarantees an approximation ratio in $o(\log \abs{\S}/ \log \log \abs{\S})$. Then, by Lemma~\ref{lm:equivalence}, there is an approximation algorithm for the Minimum Congestion Set Packing problem with approximation guarantee $o(\log \abs{U}/\log\log \abs{U})$. 

Let $\delta > 0$ be arbitrary and set $r = \lfloor n^\delta \rfloor$, where $n$ is the number of variables in the $3$-SAT instance (see Theorem~\ref{th:labelcover}). Then $r^7 2^{-\gamma r} < 1$ for all sufficiently large $n$. 
We first bound the size of the MCSP instance $(U,\X,t)$ constructed in Section~\ref{subsec:reduction}. By Lemma~\ref{lm:partition}, the size of an $(r,C)$-partition system is $|Z| = r^{O(r)} \log \abs{C}$. By Theorem~\ref{th:labelcover}, we have $\abs{C} = 2^{O(r)}$. So each set $U_h$ has cardinality at most $r^{O(r)} \cdot r = r^{O(r)}$. Also recall that the number of sets in the MCSP instance is $|\X| = \sum_{j \in [r]} |\V_j| \cdot |L_j| = n^{O(r)}$, and that the number of elements is $|U| = m = |\eset| \cdot r^{O(r)} \leq (nr)^{O(r)} = n^{O(r)} = n^{O(n^\delta)} = 2^{O(r \log r)}$. Thus
$r \geq \Omega(\log m/ \log \log m)$. 

The gap in the optimal congestion between the Yes-Instance and the No-Instance is at least $r$ (see Theorem~\ref{th:labelcover} and Lemmas~\ref{lm:YES},~\ref{lm:NO}). More precisely, for Yes-instances the congestion is at most one and for No-instances the congestion is at least $r$. Since the approximation ratio of the alleged algorithm is $o(\log m/ \log \log m)$, it is better than $r$ for all sufficiently large $n$ and hence the approximation algorithm can be used to decide the satisfiability problem. 

The running time of the algorithm is polynomial in the size of the MCSP instance, i.e., is $\mathrm{poly}(n^{O(n^\delta)}) = n^{O(n^\delta)} = 2^{O(n^{2\delta})}$. Since $\delta > 0$ is arbitrary, the theorem follows. 
\end{proof}

\section{Hardness of Robust k-Median on Line Metrics}
\label{subsec:hard:line}

We will show that the reduction from $r$-Hypergraph Label Cover to Minimum Congestion Set Packing (MCSP) can be modified to give a $\Omega(\log m/ \log \log m)$ hardness of approximation for the Robust $k$-Median problem on line metrics as well, where $m = |\S|$ is the number of client-sets. 
For this section, it is convenient to assume that the label-sets are the initial segments of the natural numbers, i.e., $L_j = \{1, \ldots, \abs{L_j}\}$ and $C = \{1, \ldots, \abs{C}\}$. 

Given a Hypergraph Label Cover instance $(G,\pi,r)$, we first construct a MCSP instance $(U,\X,t)$ in accordance with the procedure outlined in Section~\ref{subsec:reduction}. Next, from this MCSP instance, we construct a Robust $k$-Median instance $(V,\S,d)$ as described below.
\begin{itemize}
\item We create a location in $V$ for every set $X(v,\ell) \in \X$. To simplify the notation, the symbol $X(v,\ell)$ will represent both a set in the instance $(U,\X,t)$, and a location in the instance $(V,\S,d)$. Thus, we have $V = \{X(v,\ell) \in \X\}$. Furthermore, we create a set of clients $S(e)$ for every element $e \in U$, which consists of all the locations whose corresponding sets in the MCSP instance contain the element $e$. Thus, we have $\S = \{ S(e) \, : \, e \in U\}$, where $S(e) = \{ X(v,\ell) \in \X \, : \, e \in X(v,\ell)\}$ for all $e \in U$. This step is same as in Lemma~\ref{lm:equivalence}.
\item We now describe how to embed the locations in $V$ on a given line. For every vertex $v \in \V_j, j \in [r]$, the locations $X(v,1), \ldots, X(v,\abs{L_j})$ are placed next to one another in sequence, in such a way that the distance between any two consecutive locations is exactly one. Formally, this gives $d(X(v,\ell), X(v,\ell')) = |\ell' - \ell|$ for all $\ell, \ell' \in L_j$. Furthermore, we ensure that any two locations corresponding to two different vertices in $\V$ are \emph{not close to each other}. To be more specific, we have the following guarantee: $d(X(v,\ell), X(v',\ell')) \geq 2$ whenever $v \neq v'$. It is easy to verify that $d$ is a line metric. 
\item We define $k \leftarrow |\X| - t$. 
\end{itemize}

Note that as $k = |\X| -t$, there is a one to one correspondence between the solutions to the MCSP instance and the solutions to the Robust $k$-Median instance. Specifically, a set in $\X$ is picked by a solution to the MCSP instance iff the corresponding location is \emph{not} picked in the Robust $k$-Median instance.

\begin{lemma}[Yes-Instance]
\label{lm:reduction:line1}
Suppose that there is a labeling strategy $\sigma$ that strongly satisfies every edge in the Hypergraph Label Cover instance $(G,\pi,r)$. Then there is a solution to the Robust $k$-Median instance $(V,\S,d)$ with objective one.
\end{lemma}

\begin{proof}
Recall the proof of Lemma~\ref{lm:YES}. We construct a solution $\X' \subseteq \X$, $|\X'| = t$, to the MCSP instance $(U,\X,t)$ as follows. For every vertex $v\in \V_j, j \in [r]$, the solution $\X'$ contains the set $X(v,\sigma(v))$. Now, focus on the corresponding solution $F_{\X'} \subseteq V$ to the Robust $k$-Median instance, which picks a location $X$ iff $X \notin \X'$. Hence, for every vertex $v \in \V_j$, $j \in [r]$, all but one of the locations $X(v,1), \ldots, X(v,\abs{L_j})$ are included in $F_{\X'}$. Since any two consecutive locations in such a sequence are unit distance away from each other, the cost of connecting any location in $V$ to the set $F_{\X'}$ is either zero or one, i.e., $d(X,F_{\X'}) \in \{0,1\}$ for all $X \in V = \X$.

For the rest of the proof, fix any set of clients $S(e) \in \S$, $e \in U$. The proof of Lemma~\ref{lm:YES} implies that the element $e$ incurs congestion one  under $\X'$. Hence, the element belongs to exactly one set in $\X'$, say $X^*$. Again, comparing the solution $\X'$ with the corresponding solution $F_{\X'}$, we infer that $S(e) \setminus F_{\X'} = \{X^*\}$. In other words, every location in $S(e)$, except $X^*$, is present in the set $F_{\X'}$. The clients in such locations require zero cost for getting connected to $F_{\X'}$. Thus, the total cost of connecting the clients in $S(e)$ to the set $F_{\X'}$ is at most: $\sum_{X \in S(e)} d(X, F_{\X'}) = d(X^*, F_{\X'}) \leq 1.$

Thus, we see that every set of clients in $\S$ requires at most unit cost for getting connected to $F_{\X'}$. So the solution $F_{\X'}$ to the Robust $k$-Median instance indeed has objective one. 
\end{proof}

\begin{lemma}[No-Instance]
\label{lm:reduction:line2}
Suppose that every labeling weakly satisfies at most a $2^{-\gamma r}$ fraction of the edges in the Hypergraph Label Cover instance $(G,\pi,r)$, for some constant $\gamma$. Then every solution to the Robust $k$-Median instance $(V,\S,d)$ has objective at least $r$.
\end{lemma}

\begin{proof}
Fix any solution $F \subseteq V$ to the Robust $k$-Median instance $(V,\S,d)$, and let $\X'_F \subseteq \X$ denote the corresponding solution to the MCSP instance $(U,\X,t)$. Lemma~\ref{lm:NO} states that there is some element $e \in U$ with congestion at least $r$ under $\X'_F$. In other words, there are at least $r$ sets $X_1, \ldots, X_r \in \X'_F$ that contain the element $e$. The locations corresponding to these sets are not picked by the solution $F$. Furthermore, the way the locations have been embedded on a line ensures that the distance between any location and its nearest neighbor is at least one. Hence, we have $d(X_i,F) \geq 1$ for all $i \in [r]$. Summing over these distances, we infer that the total cost of connecting the clients in $S(e)$ to $F$ is at least $\sum_{i \in [r]} d(X_i, F) \geq r$. Thus, the solution $F$ to the Robust $k$-Median instance has objective at least $r$.
\end{proof}

Finally, applying Lemmas~\ref{lm:reduction:line1},~\ref{lm:reduction:line2}, and an argument similar to the proof of Theorem~\ref{th:hard:uniform}, we get the following result.

\begin{theorem}
\label{th:hard:line}
The Robust $k$-Median problem $(V,\S,d)$ is $\Omega(\log m/\log \log m)$ hard to approximate even on line metrics, where $m = |\S|$, unless $\textsf{NP} \subseteq \cap_{\delta > 0} \textsf{DTIME} (2^{n^{\delta}})$.
\end{theorem}

\section{Heuristics}\label{sec:experiments}

The Robust $k$-Median problem is a hard to approximate real-world problem and as such heuristic solutions are interesting. In this section, we complement our negative theoretical results with an evaluation of simple heuristics for the Robust $k$-Median problem. In particular we look at two greedy strategies and two variants of a local search approach. We consider a slight generalization of the problem where clients and facilities are at separate locations. This is more realistic and no easier than the original problem, as one can simply place a facility at every client position to solve an instance of the problem as defined in Definition~\ref{def:mainproblem}. Due to space constraints, the full version of this section is deferred to Appendix~\ref{app:heuristics}.

We implemented\footnote{Code and data are available at \url{http://resources.mpi-inf.mpg.de/robust-k-median/code-data.7z}} and compared the following heuristics to the LP relaxation (see \ref{exp:methods}).\smallskip

\noindent \textbf{Greedy Upwards.} Initialize all facilities as closed. Open the facility that reduces the cost maximally. Repeat until $k$ facilities are open.\smallskip

\noindent \textbf{Greedy Downwards.} Initialize all facilities as open. Close the facility that increases the cost minimally. Repeat until $k$ facilities are open.\smallskip

\noindent \textbf{Local Search.} Open $k$ random facilities. Compare all solutions that can be obtained from the current solution by closing $\ell$ facilities and opening $\ell$ facilities. Replace the current solution by the best solution found. Repeat until the current solution is a local optimum. In the experiments we use $\ell=2$.\smallskip

\noindent \textbf{Randomized Local Search.} Same as Local Search, but instead of considering \emph{all} solutions in the neighborhood, sample only a random subset. The size of the subset is an additional parameter to the heuristic. In the experiments we use $\ell=3$ and 200 random neighbors.\smallskip

We generate three kinds of 2D-instances. In the first, \emph{uniform}, the clients are uniformly distributed and all groups have the same size. The other kinds of instances cluster the client groups according to gaussian distributions. The intuition is that in real world instances client groups have something in common, e.g.\ all come from the same city. They two kinds differ in the number of clients per group. We have \emph{gauss-const} instances where all groups have the same size and \emph{gauss-exp} instances where group sizes follow an exponential distribution.

\begin{table}\centering
\begin{tabular}{lccc}\toprule
Heuristic         & Uniform     & Gauss-Const   & Gauss-Exp   \\\midrule
Greedy Up         & 1.65  (1.49)  & 5.18 (5.24) & 6.63 (5.94)\\
Greedy Down       & 1.45  (1.42)  & 2.92 (2.92) & 2.12 (2.05)\\
Local Search      & 1.13  (1.12)  & 1.63 (1.62) & 1.41 (1.39)\\
Randomized Local Search & 1.53  (1.48)  & 2.15 (2.29)   & 2.37 (2.36)\\
\bottomrule
\end{tabular}
\caption{Mean Performance as a multiple of the LP relaxation value, rounded to three digits. In parentheses we provide the median. 1654 uniform instances, 1009 Gauss-Const instances, and 2029 Gauss-Exp instances of varying sizes were solved. The reported performance is over the instances where the heuristics perform worse that the LP relaxation.}
\label{tab:uniform means}
\end{table}

Table~\ref{tab:uniform means} summarizes the results. The performance differences in Table~\ref{app:tab:uniform means} are statistically significant with a very small two-sided $p$-value, according to a Wilcoxon signed-rank test, except for the difference between Greedy Downward and Randomized Local Search on Uniform and Gauss-Const instances. In these cases the $p$-value is 0.66, respectively 0.08.

Since we use an LP relaxation as a comparison point, we do not know whether the instances where the heuristics find a worse solution are actually hard for the heuristics or whether the LP relaxation provides a much too low bound. To investigate this we had a closer look at instances where both Greedy down and Local Search perform badly. For three instances we solved the integer linear program. In these instances at least it was indeed the case that the LP relaxation yielded a bad bound. This suggests that the heuristics work even better than the numbers in Table~\ref{app:tab:uniform means} indicate.

As expected instances where the \emph{robust} nature of the Robust $k$-Median problem are not as important because groups are distributed uniformly are easier than the more realistic instances where groups form clusters. For the two better heuristics, Greedy Downwards and Local Search, also perform better on instances with uneven group sizes. Here too, one can speculate that few groups dominate the problem, and finding a solution that minimizes maximum costs becomes easier.

The good performance of these simple heuristics indicate that although the Robust $k$-Median problem is hard to approximate in the worst case, a heuristic treatment can effectively find a very good approximation.

\section{Conclusion and Future Work}

We show a logarithmic lower bound for the Robust $k$-median problem on the uniform and line metrics, implying that there is no good approximation algorithm for the problem.  
However, the empirical results suggest that real-world instances are much easier, so it is interesting to see whether incorporating real-world assumptions helps reducing the problem's complexity.

For instance, if we assume that the diameter of each set $S_i$ is at most an $\epsilon$ fraction of the diameter $\Delta = \max_{u,v} d(u,v)$ of the input instance, can we obtain a constant approximation factor? This case captures the notion of ``locality'' of the communities.
We note that in our hardness instances the diameter of each set $S_i$ is $\Delta$ for uniform metric and at least $\Delta/2$ in the line metric, so these hard instances would not arise if we have the locality assumption.  
Another interesting case is a random instance where the sets $S_i$ are randomly generated by an unknown distribution. 

One can also approach this problem from the parameterized complexity angle. 
In particular, can we obtain an $O(1)$ approximation algorithm in time $g(k) \operatorname{poly}(n)$?

\bibliography{bibliography}

\newpage

\appendix 

\section{Hypergraph Label Cover}
\label{sec:labelcover}
An instance of $r$-Hypergraph Label Cover is equivalent to the $r$-Prover system as used by Feige~\cite{feige} in proving the hardness of approximation for Set Cover.
We discuss the equivalence in this section.

In the $r$-prover system, there are $r$ provers $P_1,\ldots, P_r$ and a verifier $V$. 
Each prover is associated with a codeword of length $r$ in such a way that the hamming distance between any pair $P_i, P_j$ is at least ${\sf ham}(P_i, P_j) = r/2$; this is possible if $r$ is a power of two because we can use Hadamard code. Given an input 3-SAT formula $\phi$, the verifier selects $r$ clauses uniformly and independently at random. 
Call these clauses $C_1,\ldots, C_r$. 
From each such clause, the verifier selects a variable uniformly and independently at random. 
These variables are called $x_1,\ldots, x_r$.
Prover $P_i$ receives a clause $C_j$ if the $j$th bit of its codeword is $0$; otherwise, it receives variable $x_j$.
The property of Hadamard code guarantees that each prover would receive $r/2$ clauses and $r/2$ variables. 

Then each prover $P_i$ is expected to give an assignment to all involved variables it receives and sends this assignment to the verifier.  
The verifier then looks at the answers from $r$ provers and has two types of acceptance predicates. 
\begin{itemize} 
\item (Weak acceptance) At least one pair of answers is consistent. 

\item (Strong acceptance) All pairs of answers are consistent.  
\end{itemize}   

Applying parallel repetition theorem~\cite{Raz98}, Feige argues the following. 

\begin{theorem} (\cite[Lemma 2.3.1]{feige})
If $\Phi$ is a satisfiable 3-SAT(5) formula, then there is provers' strategy that always causes the verifier to accept. 
Otherwise, the verifier weakly accepts with probability at most $r^2 2^{- \gamma r}$ for some universal constant $\gamma >0$.
\end{theorem} 

Now we show how Theorem~\ref{th:labelcover} follows by constructing the instance of Hypergraph Label Cover $(V,E)$ based on the $r$-prover system.  
For each prover $j$, we create a set $V_j$ consisting of vertices $v$ that correspond to possible query sent to prover $j$, so we have $|V_j| = (5n/3)^{r/2} n^{r/2}$.   
For each possible random string $x$, we have an edge $h_x$ that contains $r$ vertices, corresponding to queries sent to the provers. 
It can be checked that the total number of possible random strings is $(5n)^r$, and the degree of each vertex is $3^{r/2} 5^{r/2} = 15^{r/2}$; notice that this is equal to $r |E|/|V|$.  
A prover strategy corresponds to the label of vertices, and the acceptance probability is exactly the fraction of satisfied edges.  
Moreover, for each possible query, the number of possible answers is at most $7^r$ (for each clause, there are $7$ ways to satisfy it). This implies that $|L_j| \leq 7^r$.


\section{Heuristics}\label{app:heuristics}

The Robust $k$-Median problem is a real-world problem and as such needs to be solved as well as possible despite its hardness of approximation. In this section, we complement our negative theoretical results with an experimental evaluation of different simple heuristics for the Robust $k$-Median problem. In particular we look at two variants of a greedy strategy and two variants of a local search approach. We consider a slight generalization of the problem where clients and facilities are at separate locations. This is more realistic and no easier than the original problem, as one can simply place a facility at every client position to solve an instance of the problem as defined in Definition~\ref{def:mainproblem}.

This is by no means an exhaustive exploration of the possible solution space. However, the results we obtain indicate that a heuristic treatment of the Robust $k$-Median problem can yield surprisingly good solutions, even if the heuristics are very naive.

For our experiments we consider instances in the plane, as these are closest to the real-world motivation for the problem. We wanted to check how the structure of the instance influences the performance of the heuristics. We suspected that instances where the clients are distributed uniformly are easy, as intuitively a solution that is good for one group of clients is good for all groups. 

The robust version of the k-median problem is considered because often the exact set of clients is not known before choosing facility locations and one wants to perform well even if the worst set of possible clients turns out to be realized. It is reasonable to assume that every group of clients has something in common, for example that they come from a similar region, like a city. Therefore more realistic instances for the Robust $k$-Median problem have the groups form clusters in space. We also generate such instances for testing our heuristics.

\subsection{Methods}\label{exp:methods}

Since solving Robust $k$-Median instances to optimality is infeasible for the instances we consider\footnote{We attempted solving three instances optimally, see Figure~\ref{fig:hard instances}, but gave up on the third after nearly half a year of CPU time was consumed.}, we compare the performance of the various heuristics to the value of a LP-relaxation. We have a variable $x_j$ for each possible median location and variables $y_{ij}$ that indicate whether client $i$ is served by facility $j$. The LP is then as follows.
\begin{align*}
\min\quad & T\\
s.t.\quad & y_{ij} - x_j \le 0 & \forall i,j\\
& \sum_{j} y_{ij} \ge 1 & \forall i\\
& \sum_{i\in g} d(i,j) \cdot y_{ij} \le T & \forall \text{ groups of clients } g\\
& \sum_j x_j \le k \quad\text{and}\quad 0\le x_j \le 1 & \forall j\\
& 0\le y_{ij} \le 1 &\forall i,j
\end{align*}
To solve the LP we use the Gurobi solver~\cite{gurobi}, version 5.5.0, on a 64-bit Linux system.

Note that the assignment of the $y_{ij}$ variables is immediately clear from the assignment of the $x_j$. For location $i$, let $j_1$, $j_2$, \ldots $j_n$ be the locations ordered by increasing distance. Then $y_{i j_\ell} = \min(x_{j_\ell},1 - (y_{ij_1} + \ldots + y_{i j_{\ell - 1}}))$. The constraint $y_{i j_\ell} \le \min(\, ,\, )$ is already expressed by the first two constraints. It could however be put into the objective via the big $M$-method. Consider a minimization problem $\min T$ subject to $x = \min(b,c)$. Let $M$ be large integer and consider $\min T + M t$ subject to $x \le b$, $x \le c$, $t \le b-x$, and $t \le c -x$. Observe that $t = \min(b,c) - x$ in an optimal solution. One needs to choose $M$ big enough so that $t$ must be zero in an optimal vertex solution. It is however unclear whether this will speed up the solution. We have not tried this method.

We implemented and compared the following heuristics:\medskip

\noindent \textbf{Greedy Upwards.} Initialize all facilities as closed. Open the facility that reduces the cost maximally. Repeat until $k$ facilities are open.\medskip

\noindent \textbf{Greedy Downwards.} Initialize all facilities as open. Close the facility that increases the cost minimally. Repeat until $k$ facilities are open.\medskip

\noindent \textbf{Local Search.} Open $k$ random facilities. Compare all solutions that can be obtained from the current solution by closing $\ell$ facilities and opening $\ell$ facilities. Replace the current solution by the best solution found. Repeat until the current solution is a local optimum. In the experiments we use $\ell=2$.\medskip

\noindent \textbf{Randomized Local Search.} Same as Local Search, but instead of considering \emph{all} solutions in the neighborhood, sample only a random subset. The size of the subset is an additional parameter to the heuristic. In the experiments we use $\ell=3$ and 200 random neighbors.\bigskip

Note that taking the solution of one of the greedy algorithms as starting point for a local search is an obvious improvement, but this would prevent us from comparing the local search algorithm with the greedy heuristic.

The local search heuristic is closely related to Lloyd's algorithm for the k-means problem. In Lloyd's algorithm, a random set of centers is chosen and iteratively updated by moving the centers to the centroids of the clients that fall in their voronoi cell. This improves the total distance from the centers to all clients in every iteration.

In our setting, we want to reduce the cost of the group of clients that currently incurs the maximal cost. This can be done by moving a facility closer to this group of clients, that is, closing one facility and opening another that reduces the objective function. The local search algorithm, by closing and opening more than one facility at a time, does this at least as well.

We create instances in the plane and use the euclidean distance. We create two types of instances. In the first type the clients and facilities are uniformly distributed in a $100\times 100$ square. We call these instances the \emph{uniform} instances. In these instances all groups of clients contain the same number of clients. The $k$ we use for the experiments is 7.

The second kind uses random gaussian distributions to sample client positions. To generate the gaussian distributions we sample a matrix $M$ with $v_1,v_2$ on the diagonal, where the two values are chosen uniformly at random from $[0,50]$, the matrix is then rotated by a uniformly random angle. The result is the covariance matrix of the gaussian distribution. The mean is a random point in a $100\times 100$ square. These instances we call \emph{gauss}. We generate two subgroups of instances, in the first subgroup, \emph{gauss-const}, all groups of clients have the same number of clients, in the second subgroup, \emph{gauss-exp}, the number of clients in a group is sampled from an exponential distribution. Figure~\ref{fig:instances} shows examples for the different kind of instances we generate.
\begin{figure}
\centering
\subfloat[Uniform]{
\includegraphics[width=0.31\linewidth]{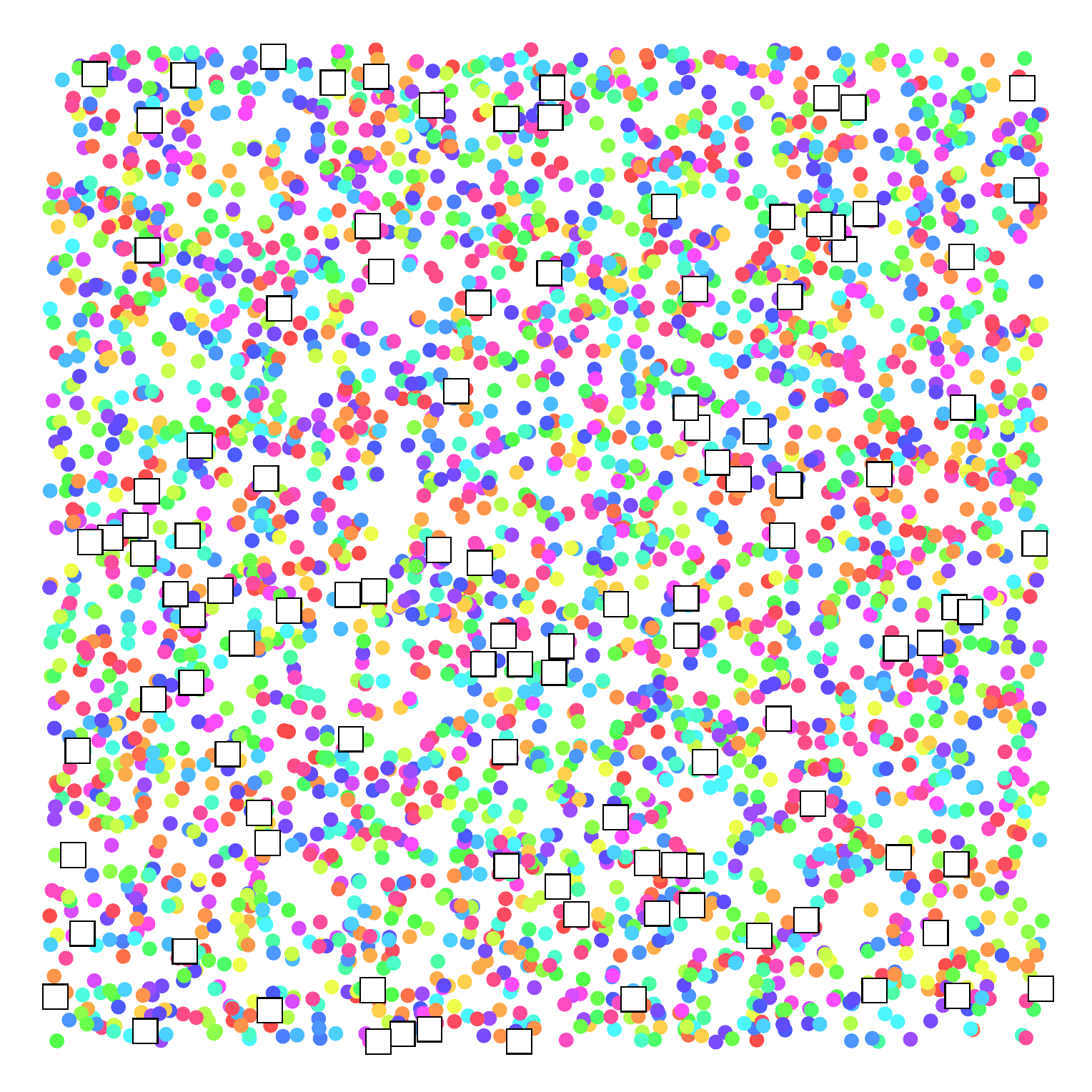}\label{fig:uniform}}
\subfloat[Gauss-Const]{\includegraphics[width=0.31\linewidth]{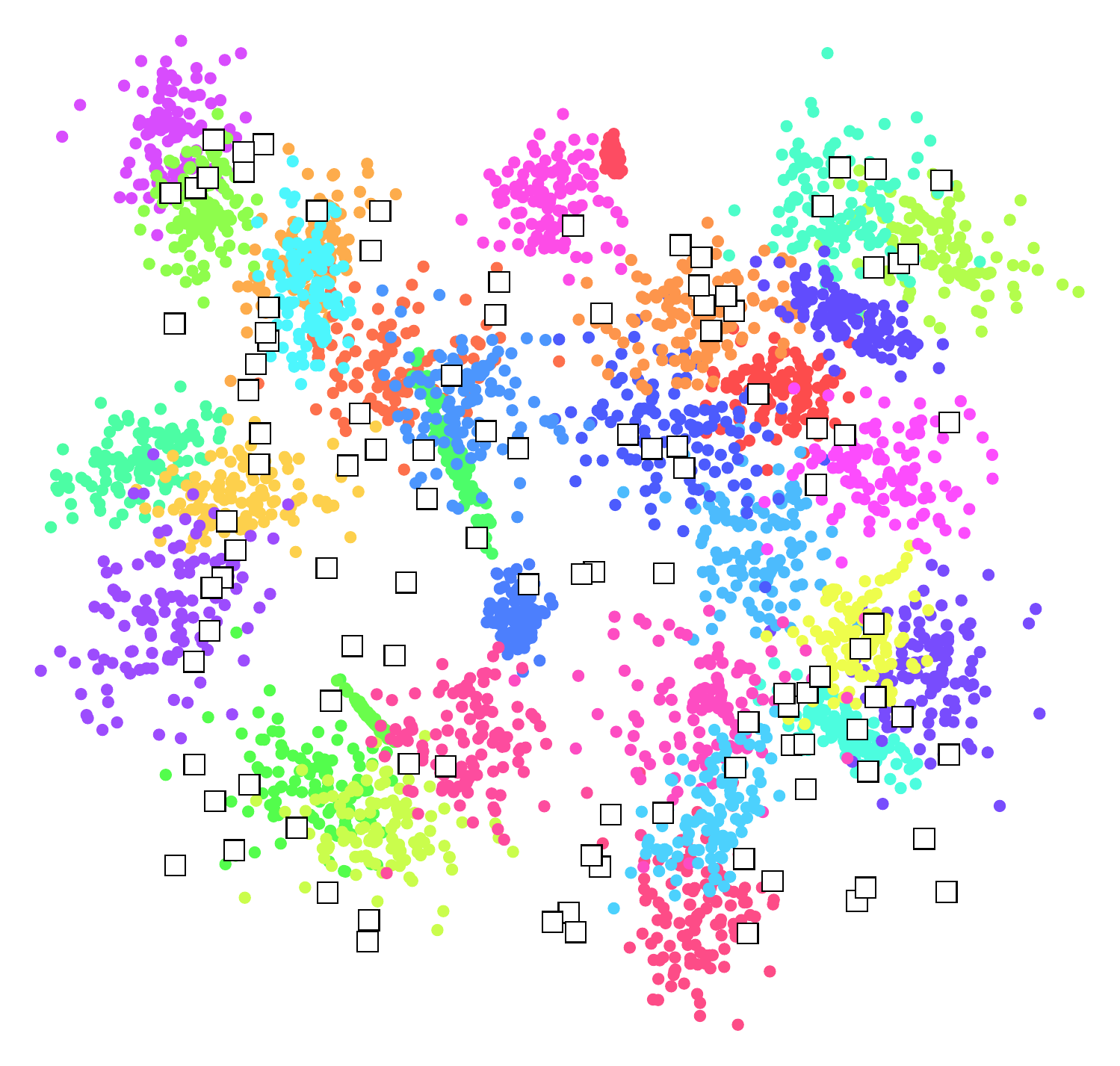}\label{fig:gauss-const}}
\subfloat[Gauss-Exp]{\includegraphics[width=0.31\linewidth]{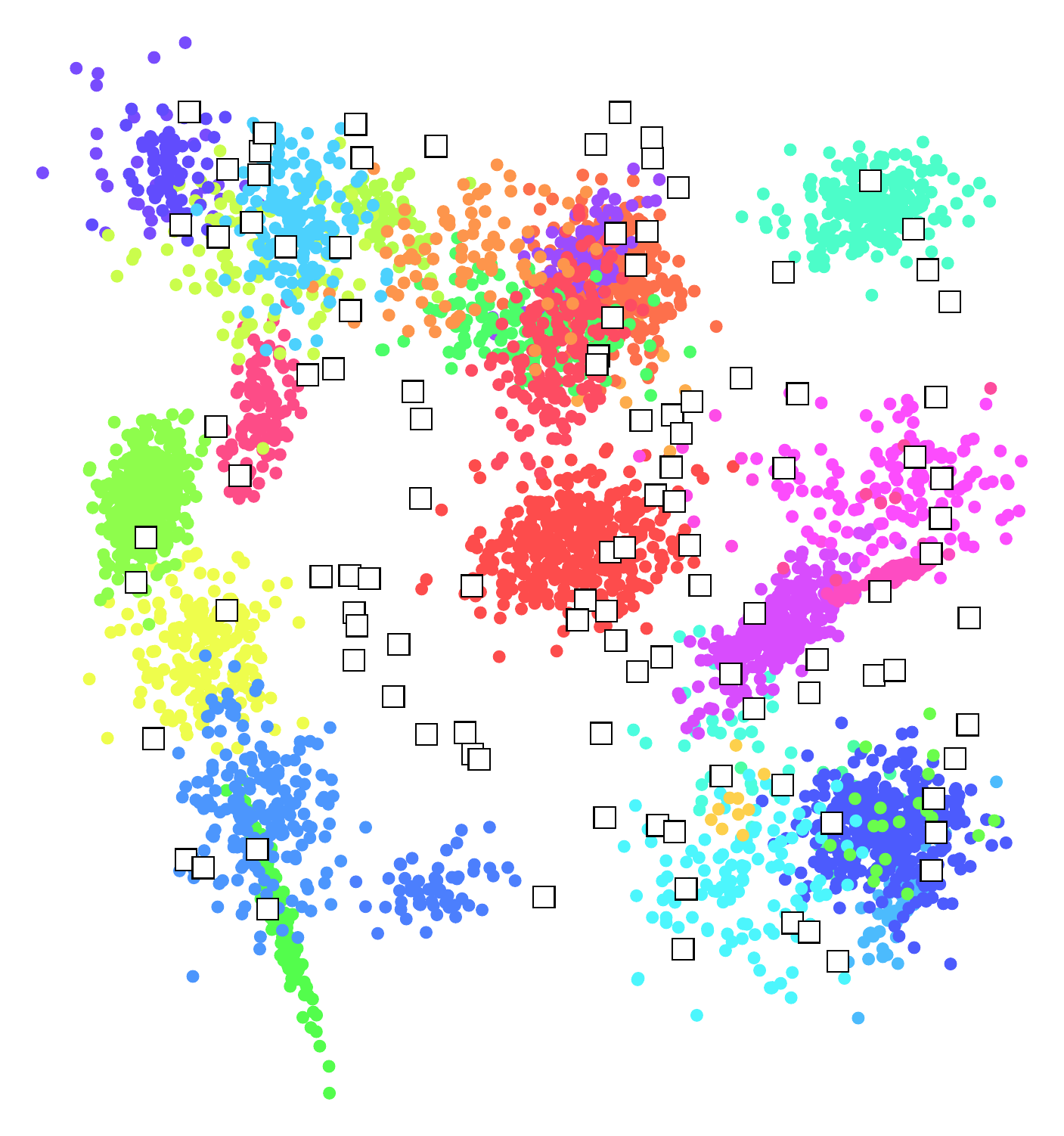}\label{fig:gauss-exp}}
\caption{Examples for the kind of instances we generate. Circles are clients, squares are facilities, colors indicate group membership.}
\label{fig:instances}
\end{figure}

As we didn't put much effort into optimizing our heuristics for speed (for example we don't use spatial search structures to find nearest neighbors), we don't report execution time and focus solely on solution quality. Nevertheless it is clear that the greedy strategies are much simpler to implement and much faster than the local search heuristics.

We report average performance on instances where the solution is worse than the LP value, as small, easy instances otherwise skew the results. To conclude relative performance advantages between heuristics we use a Wilcoxon signed-rank test as implemented in SciPy 0.12.0.

All computer code we wrote to run the experiments and analyze the results, as well as the instances we solved, is available online at \url{http://resources.mpi-inf.mpg.de/robust-k-median/code-data.7z}.

\subsection{Results}

\begin{table}\centering
\begin{tabular}{lccc}\toprule
Heuristic         & Uniform     & Gauss-Const   & Gauss-Exp   \\\midrule
Greedy Up         & 1.65  (1.49)  & 5.18 (5.24) & 6.63 (5.94)\\
Greedy Down       & 1.45  (1.42)  & 2.92 (2.92) & 2.12 (2.05)\\
Local Search      & 1.13  (1.12)  & 1.63 (1.62) & 1.41 (1.39)\\
Randomized Local Search & 1.53  (1.48)  & 2.15 (2.29)   & 2.37 (2.36)\\
\bottomrule
\end{tabular}
\caption{Mean Performance as a multiple of the LP relaxation value, rounded to three digits. In parentheses we provide the median. 1654 uniform instances, 1009 Gauss-Const instances, and 2029 Gauss-Exp instances of varying sizes were solved. The reported performance is over the instances where the heuristics perform worse that the LP relaxation.}
\label{app:tab:uniform means}
\end{table}
\begin{table}\centering
\subfloat[Uniform]{\begin{tabular}{lcccccccccc}
\toprule
Clients & \multicolumn{10}{c}{Facilities}\\ \cmidrule(r){1-1}\cmidrule(l){2-11}
  & \multicolumn{2}{c}{10} & \multicolumn{2}{c}{110} & \multicolumn{2}{c}{210} & \multicolumn{2}{c}{310} & \multicolumn{2}{c}{410} \\ \cmidrule(lr){2-3}\cmidrule(lr){4-5}\cmidrule(lr){6-7} \cmidrule(lr){8-9} \cmidrule(lr){10-11}
  & Greedy & Search & GD & LS & GD & LS & GD & LS & GD & LS\\
10  & 1.00&1.00&  1.12&1.00&  1.31&1.01&  1.39&1.02&  1.40&1.01  \\
160 & 1.01&1.01&  1.6&1.17&  1.63&1.17&  1.68&1.15&  1.63&1.15 \\
310 & 1.01&1.01&  1.64&1.21&  1.69&1.19&  1.70&1.19&  1.75&1.18 \\
460 & 1.01&1.01&  1.68&1.22&  1.73&1.21&  1.71&1.21&  1.73&1.21 \\\midrule
110 & 1.00&1.00&  1.17&1.01&  1.22&1.01&  1.25&1.01&  1.24&1.01\\
1760 & 1.0&1.0&  1.28&1.06&  1.33&1.06&  1.34&1.06&  1.34&1.06 \\
3410 & 1.0&1.0&  1.3&1.07&  1.33&1.07   \\\bottomrule
\end{tabular}}

\subfloat[Gauss-Const]
{\begin{tabular}{lcccccccccc}
\toprule
Clients & \multicolumn{10}{c}{Facilities}\\ \cmidrule(r){1-1}\cmidrule(l){2-11}
  & \multicolumn{2}{c}{10} & \multicolumn{2}{c}{110} & \multicolumn{2}{c}{210} & \multicolumn{2}{c}{310} & \multicolumn{2}{c}{410} \\ \cmidrule(lr){2-3}\cmidrule(lr){4-5}\cmidrule(lr){6-7} \cmidrule(lr){8-9} \cmidrule(lr){10-11}
  & Greedy & Search & GD & LS & GD & LS & GD & LS & GD & LS\\
10  & 1.0&1.0&  1.0&1.0&  1.0&1.0&  1.0&1.0& 1.0&1.0  \\
160 & 1.0&1.0&  2.74&1.64&  3.05&1.6&  3.33&1.62&  3.33&1.57  \\
310 & 1.0&1.0&  2.76&1.70&  3.07&1.66&  3.32&1.64\\\midrule
110 & 1.0&1.0&  1.0&1.0  & 1.0$^*$&1.0$^*$&  \\
3410 & 1.01&1.0&  2.74&1.65&  3.02$^*$&1.63$^*$\\\bottomrule
\end{tabular}}

\subfloat[Gauss-Exp]
{\begin{tabular}{lcccccccccc}
\toprule
Clients & \multicolumn{10}{c}{Facilities}\\ \cmidrule(r){1-1}\cmidrule(l){2-11}
  & \multicolumn{2}{c}{10} & \multicolumn{2}{c}{110} & \multicolumn{2}{c}{210} & \multicolumn{2}{c}{310} & \multicolumn{2}{c}{410} \\ \cmidrule(lr){2-3}\cmidrule(lr){4-5}\cmidrule(lr){6-7} \cmidrule(lr){8-9} \cmidrule(lr){10-11}
  & Greedy & Search & GD & LS & GD & LS & GD & LS & GD & LS\\
10 & && && 1.0$^*$&1.0$^*$&  1.0&1.0&  1.0&1.0\\
110 & 1.0&1.0&  1.34&1.16&  1.66&1.28&  1.65&1.26&  1.91&1.34\\
210 & 1.0&1.0&  1.9&1.41&  2.14&1.45&  2.31&1.46&  2.46&1.49\\
310 & 1.0&1.0&  2.23&1.48&  2.6&1.48&  2.69&1.51&  2.78&1.50\\\midrule
110 & 1.0&1.0&  1.0&1.0&  1.0&1.0&  1.0&1.0&  1.01&1.0\\
1210 & 1.0&1.0&  1.38&1.21&  1.56&1.23&  1.73&1.29&  1.77&1.29\\
2310 & 1.0&1.0&  1.94&1.38&  2.09&1.44&  2.48&1.41&  2.29&1.44\\
3410 & 1.0&1.0&  2.17&1.51&  2.48&1.48&  2.8&1.55&  \\\bottomrule
\end{tabular}}

\caption{Performance depending on instance size for the Greedy Downwards and Local Search heuristics. All values are averages over 50 instances, except for those marked by $^*$. For Gauss-Exp instances the number of clients is the mean of the exponential distribution times the number of groups. Values above the horizontal line come from instances with 10 clients per group, below the line instances have 110 clients per group.}
\label{app:tab:detailed experiments}
\end{table}

Table~\ref{app:tab:uniform means} summarizes the results of the experiments, Table~\ref{app:tab:detailed experiments} shows the performance for the different instance sizes for the Greedy Upwards and the Local Search heuristic. The performance differences in Table~\ref{app:tab:uniform means} are statistically significant with a very small two-sided $p$-value, except for the difference between Greedy Downward and Randomized Local Search on Uniform and Gauss-Const instances. In these cases the $p$-value is 0.66, respectively 0.08.

Since we use an LP relaxation as a comparison point, we do not know whether the instances where the heuristics find a worse solution are actually hard for the heuristics or whether the LP relaxation provides a much too low bound. To investigate this we had a closer look at instances where both Greedy down and Local Search perform badly. For three instances we attempted to solve the integer linear program and succeeded for two of them. In Figure~\ref{fig:hard instances} we see different solutions. For these instances at least it was indeed the case that the LP relaxation yielded a bad bound. This suggests that the heuristics work even better than the numbers in Table~\ref{app:tab:uniform means} indicate.

\begin{figure}
\centering
\subfloat[Uniform: LP value 2806.4, Greedy value 5982.39, Local Search value 3426.43, OPT 3230.19.]{\label{fig:hard uniform}
	\includegraphics[width=0.24\linewidth]{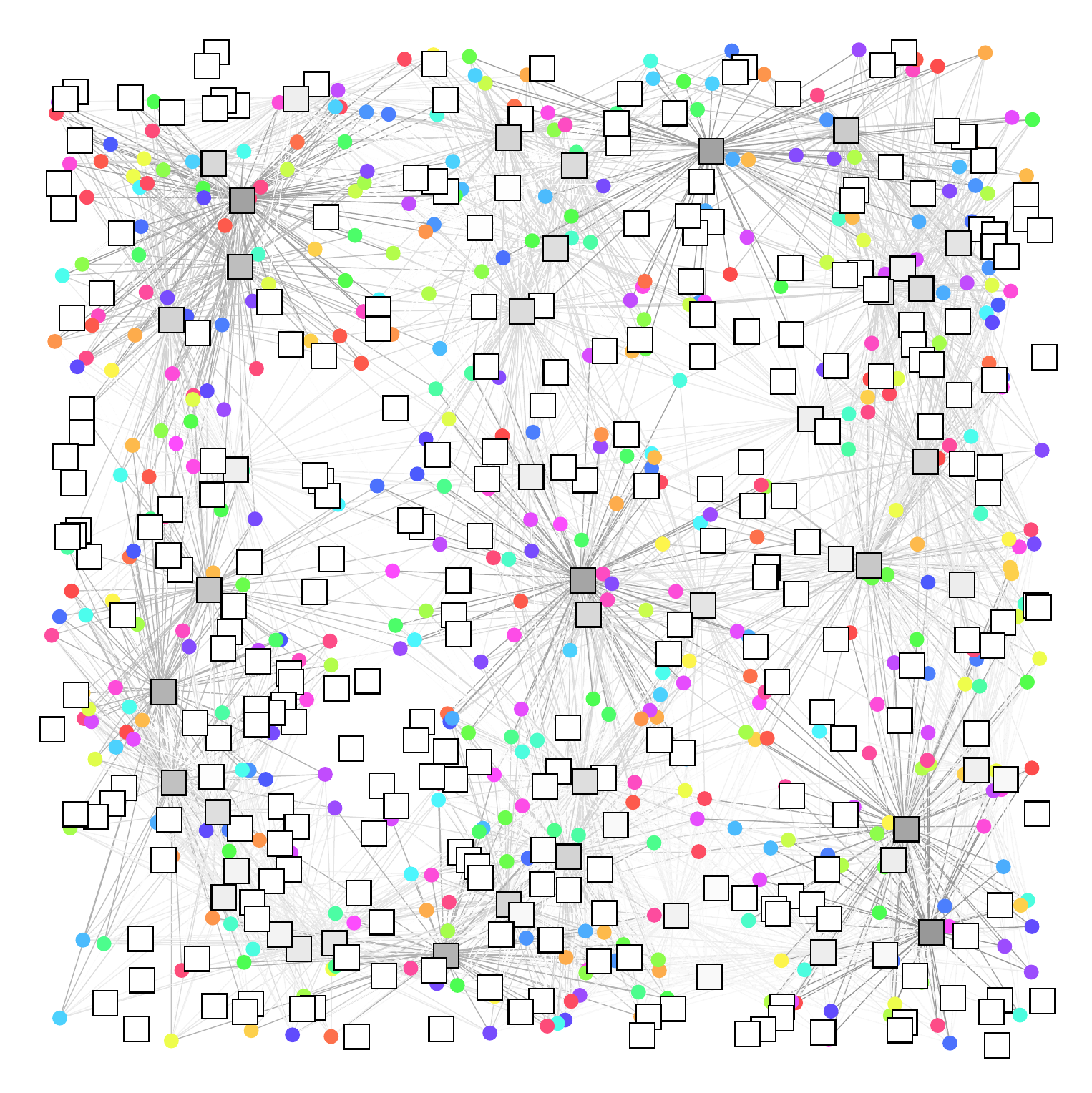}
	\includegraphics[width=0.24\linewidth]{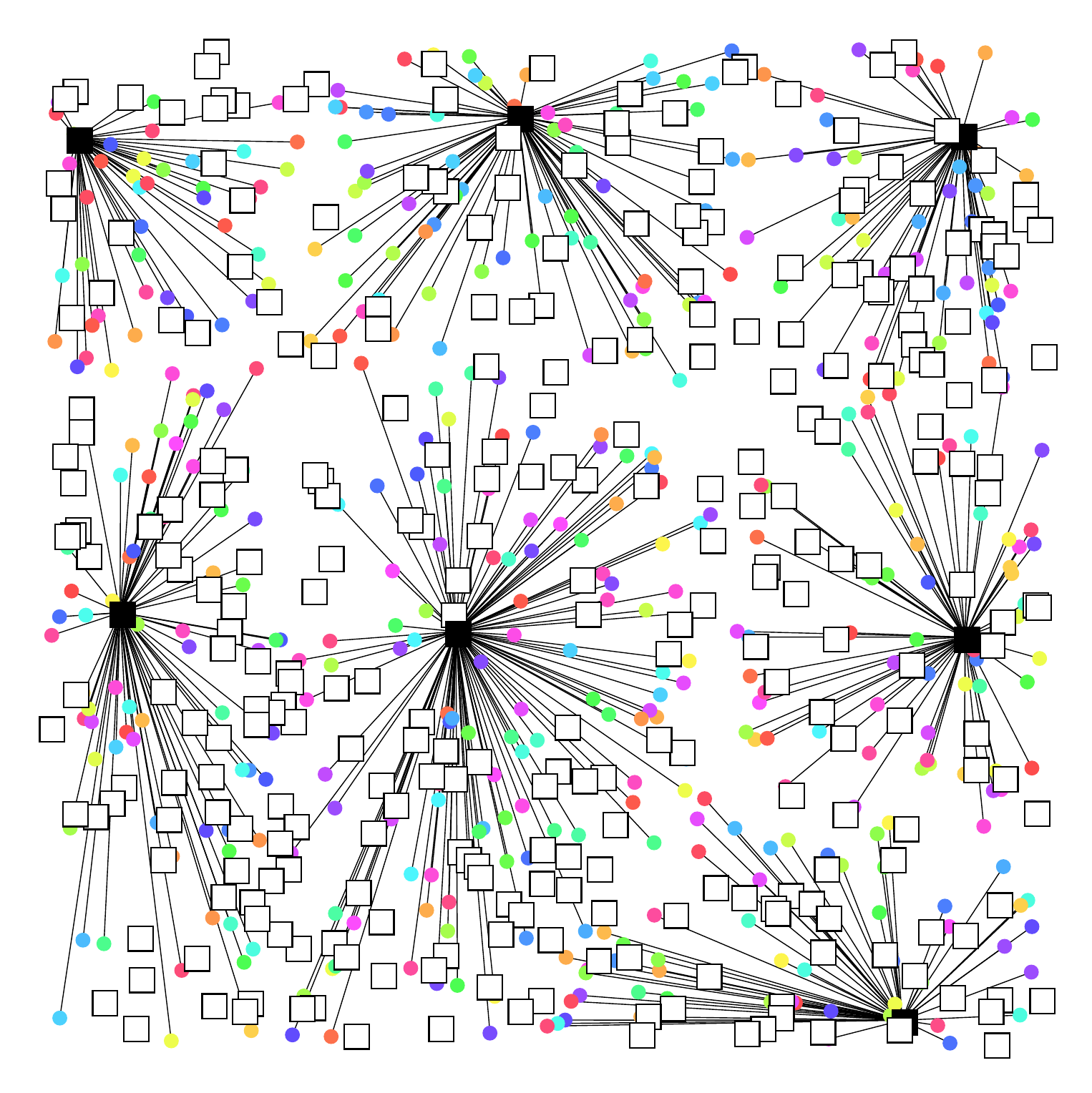}
	\includegraphics[width=0.24\linewidth]{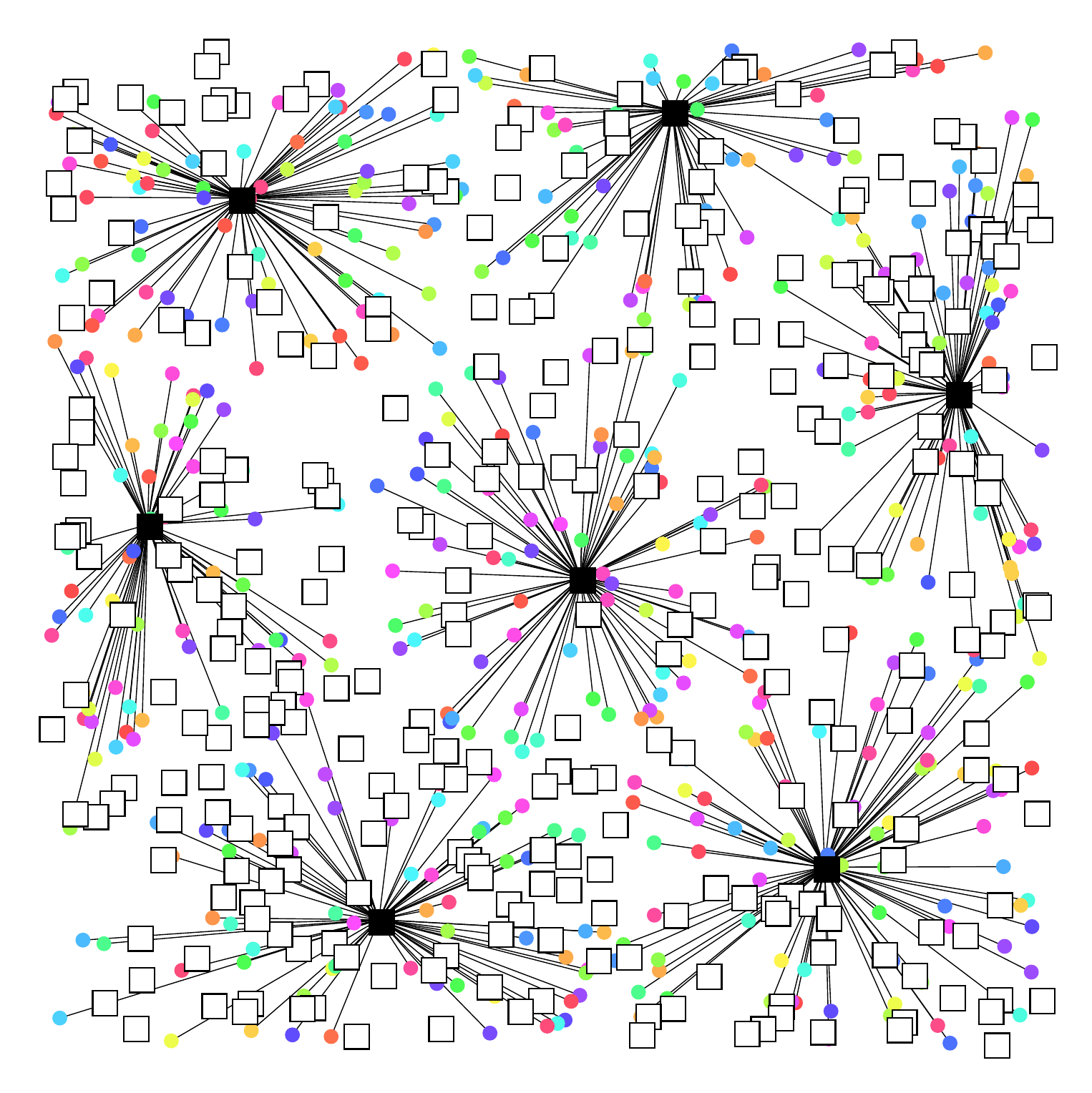}
	\includegraphics[width=0.24\linewidth]{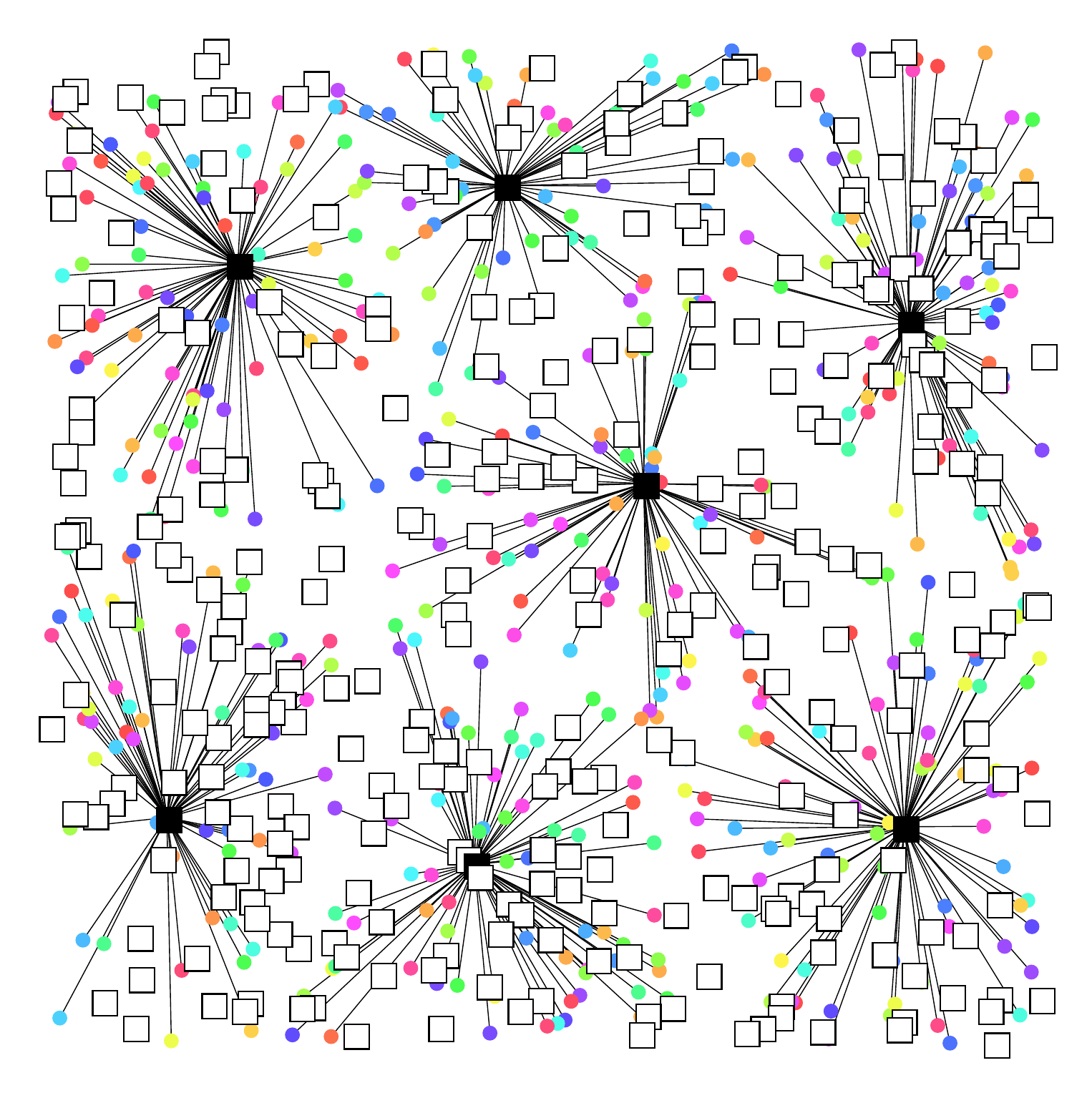}
}\\
\subfloat[Gauss-Const: LP value 1360.26, Greedy value 8307.48, Local Search value 2541.21, OPT 2505.26]{\label{fig:hard gauss const}
	\includegraphics[width=0.24\linewidth]{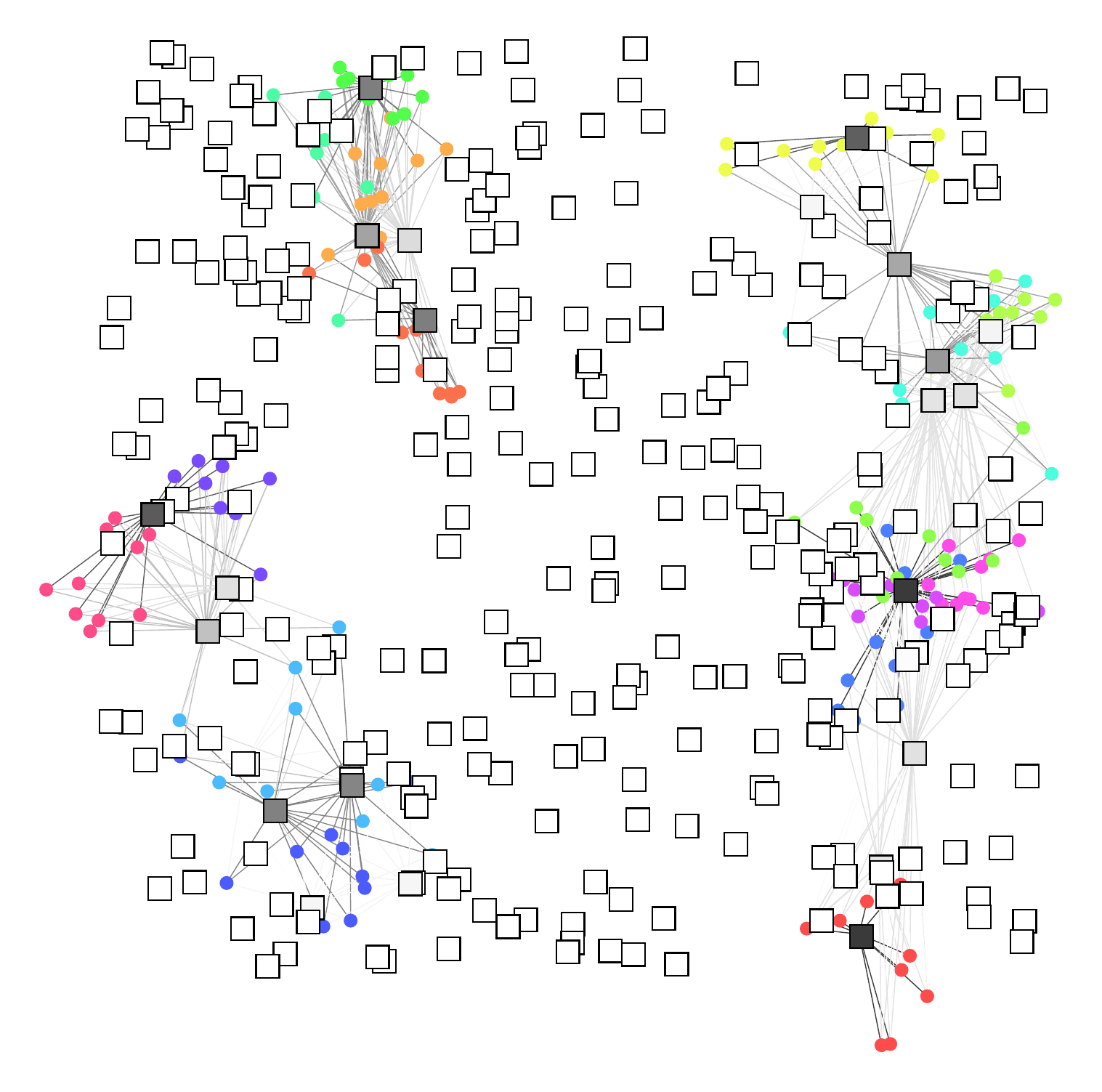}
	\includegraphics[width=0.24\linewidth]{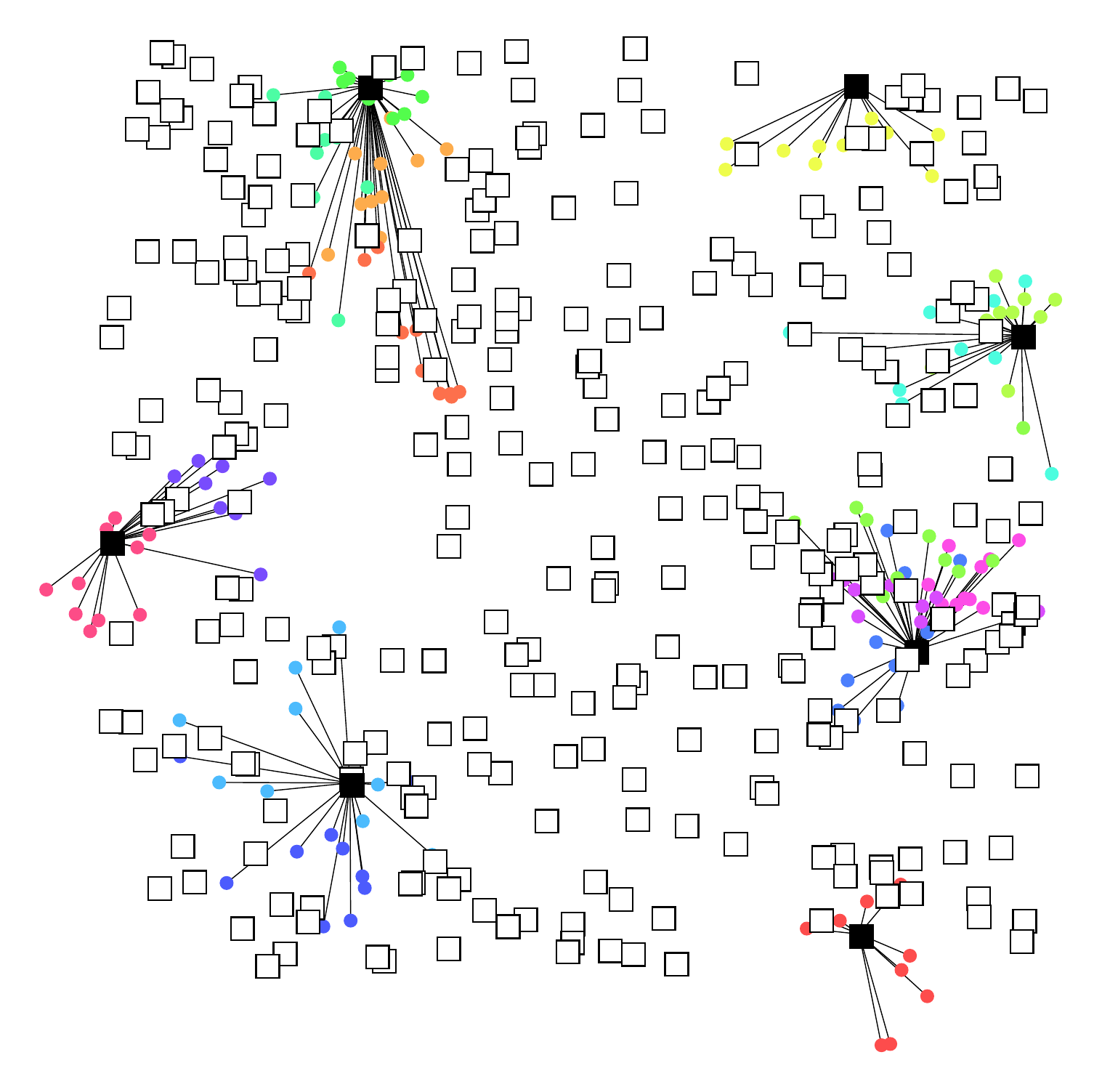}
	\includegraphics[width=0.24\linewidth]{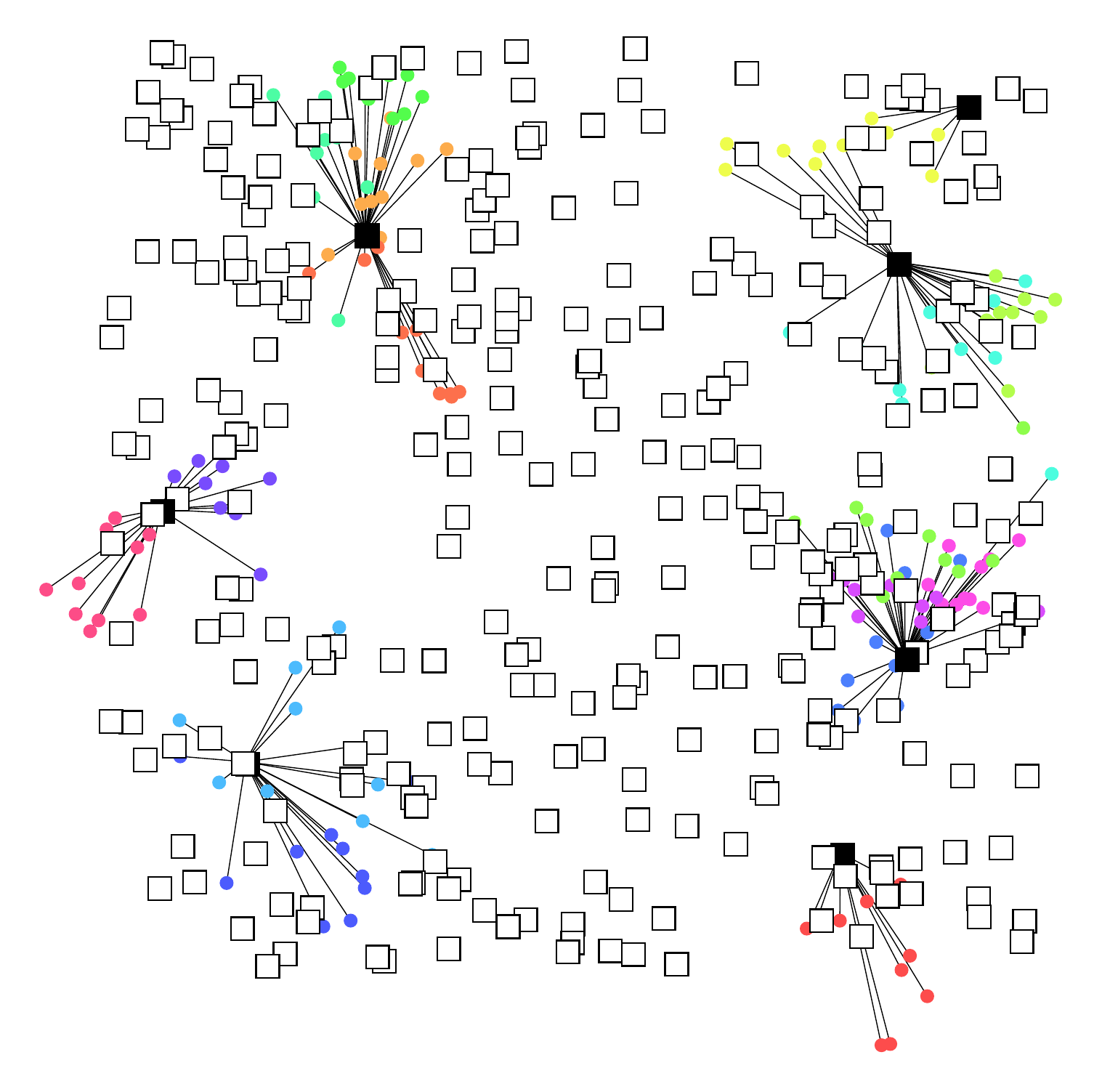}
	\includegraphics[width=0.24\linewidth]{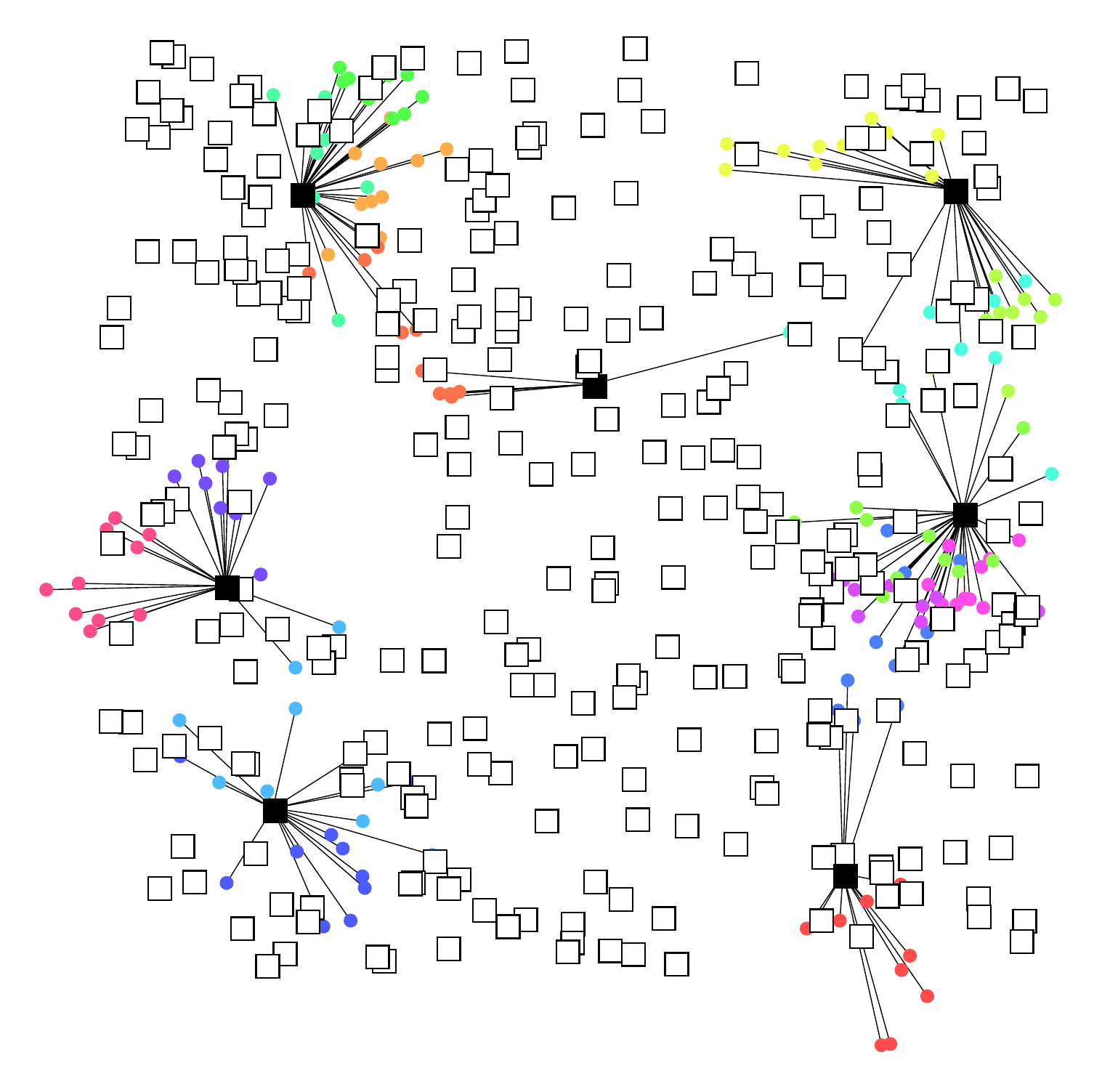}
}\\
\subfloat[Gauss-Exp: LP value 2362.06, Greedy value 10624.4, Local Search value 4354.54, 4192.31 $\le$ OPT $\le$ 4354.54]{\label{fig:hard gauss exp}
	\includegraphics[width=0.24\linewidth]{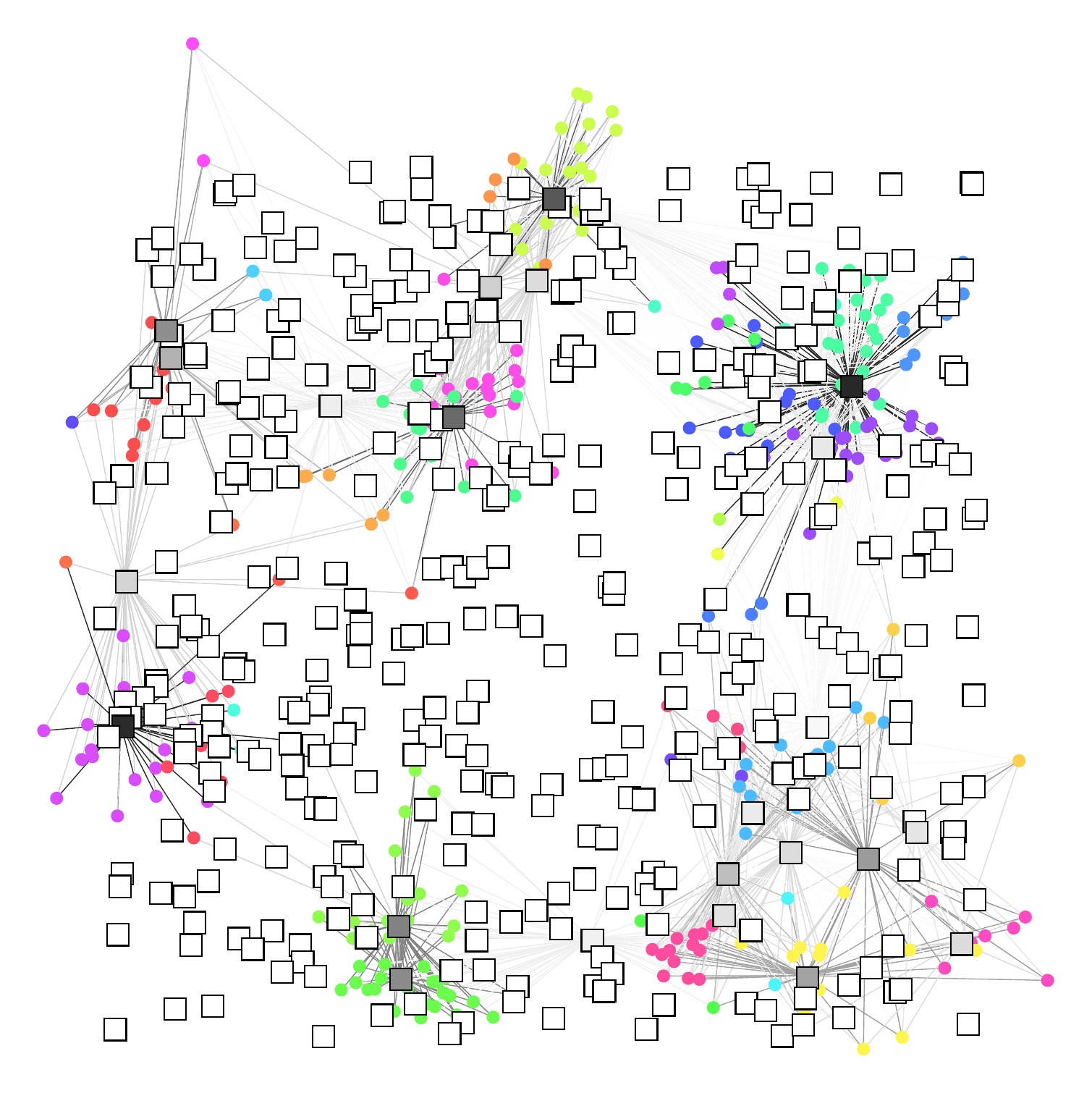}
	\includegraphics[width=0.24\linewidth]{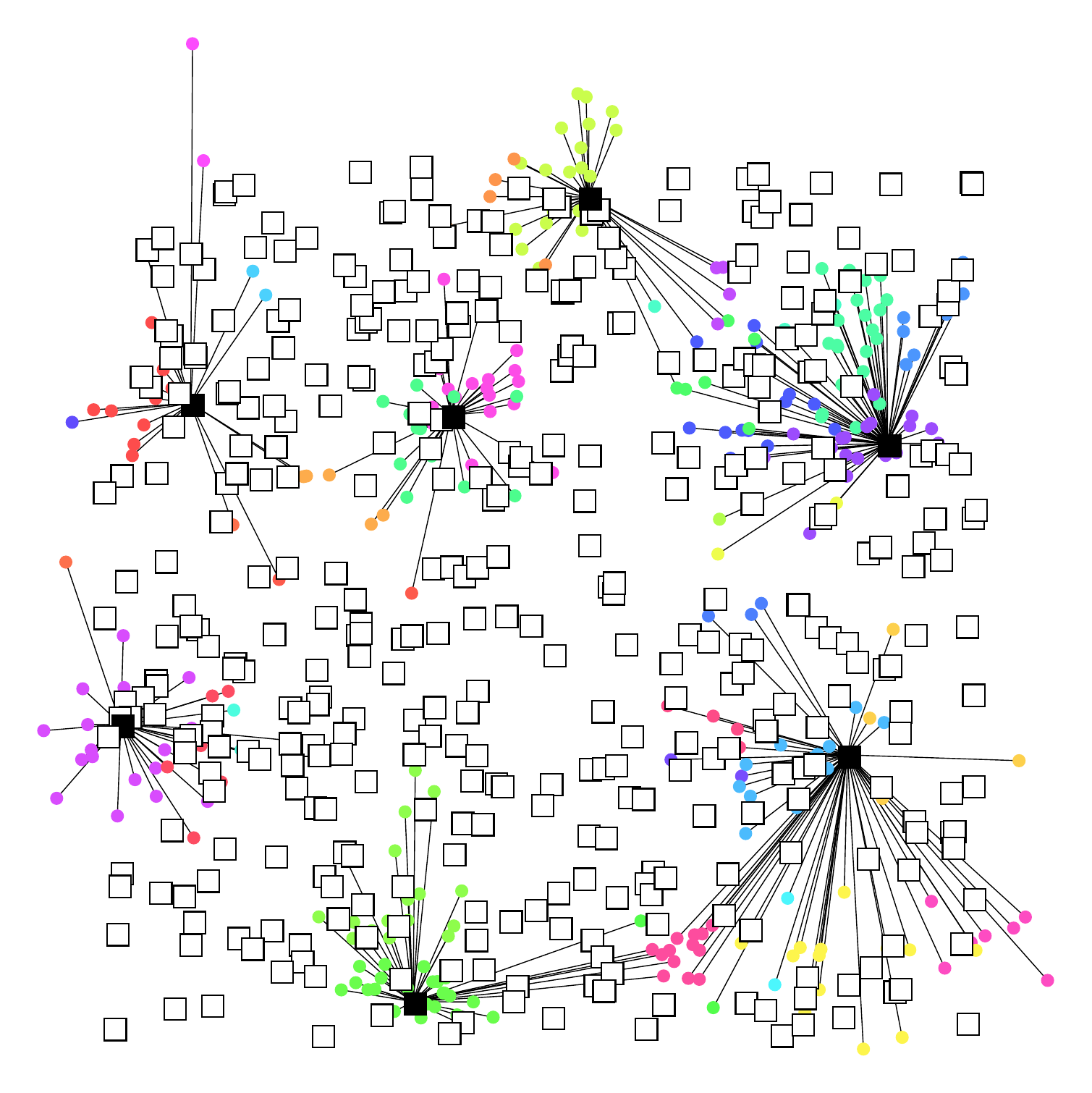}
	\includegraphics[width=0.24\linewidth]{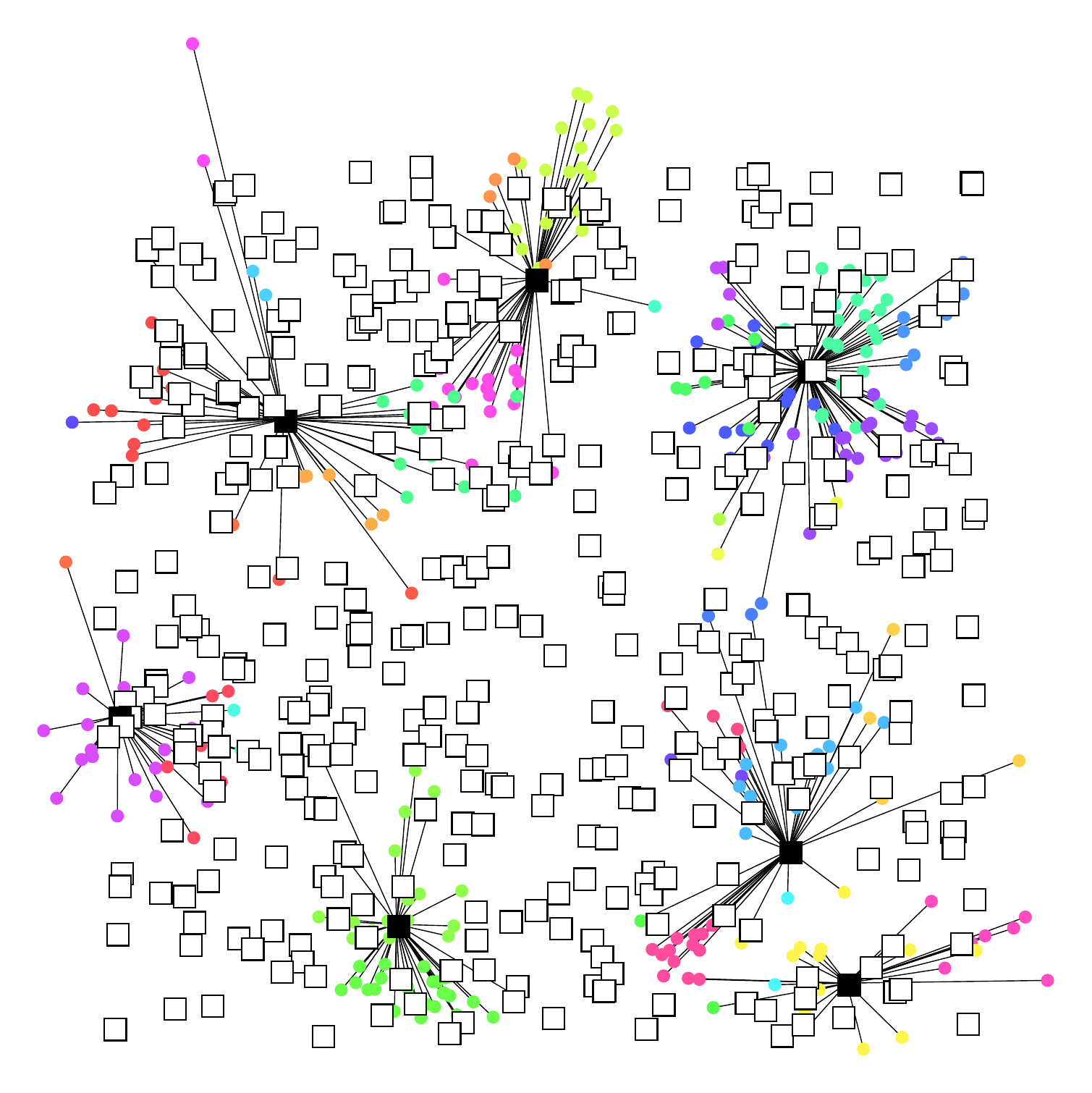}\hspace{0.24\linewidth}
}
\caption{Solutions of the different algorithms on particularly hard instances. From left to right, the LP solution, the Greedy downwards solution, the Local Search solution and the ILP solution. Darkness of facilities indicates "how open" they are in the LP relaxation. In \ref{fig:hard gauss exp} the ILP solver was still running at the time of submission, after having consumed 177 days of CPU time and 46GB of memory.}
\label{fig:hard instances}
\end{figure}

\subsection{Conclusion}

Note that all heuristics perform very well on the instances we tried. In accordance with our theoretical results, increasing the number of groups makes the instances harder, more so that increasing the number of facilities or the number of clients.

As expected instances where the \emph{robust} nature of the Robust $k$-Median problem are not as important because groups are distributed uniformly are easier than the more realistic instances where groups form clusters. For the two better heuristics, Greedy Downwards and Local Search, also perform better on instances with uneven group sizes. Here too, one can speculate that few groups dominate the problem, and finding a solution that minimizes maximum costs becomes easier.

The good performance of these simple heuristics indicate that although the Robust $k$-Median problem is hard to approximate in the worst case, a heuristic treatment can effectively find a very good approximation.

\end{document}